\documentclass{article}
\usepackage{amssymb}

\usepackage{graphicx}
\usepackage{amsmath}

\newtheorem{theorem}{Theorem}

\newtheorem{corollary}[theorem]{Corollary}

\newtheorem{definition}[theorem]{Definition}

\newtheorem{lemma}[theorem]{Lemma}

\newtheorem{problem}[theorem]{Problem}

\newtheorem{remark}[theorem]{Remark}

\newtheorem{partial solution}[theorem]{Partial Solution}

\newenvironment{proof}[1][Proof]{\textbf{#1.} }{\ \rule{0.5em}{0.5em}}
\input{tcilatex}

\begin{document}

\begin{center}
{\Huge Predicative proof theory of PDL and basic applications} {\Huge %
\medskip }

{\large L. Gordeev\smallskip }

\textit{T\"{u}bingen\medskip\ University}

l\texttt{ew.gordeew@uni-tuebingen.de}
\end{center}

\section{Extended abstract}

Propositional dynamic logic (\textbf{PDL}) was derived by M. J. Fischer and
R. Ladner \cite{FischerLadner0}, \cite{FischerLadner} from dynamic logic
where it plays the role that classical propositional logic plays in
classical predicate logic. Conceptually, it describes the properties of the
interaction between programs (as modal operators) and propositions that are
independent of the domain of computation. The semantics of \textbf{PDL} is
based on Kripke frames and comes from that of modal logic. Corresponding
sound and complete Hilbert-style formalism was proposed by K. Segerberg \cite
{Seger} (see also \cite{Par}, \cite{DynamicLogic}). Gentzen-style treatment
is more involved. This is because the syntax of \textbf{PDL} includes
starred programs $P^{\ast }$ which make finitary sequential formalism
similar to that of (say) Peano Arithmetic with induction (\textbf{PA}) that
allows no full cut-elimination. In the case of \textbf{PA}, however, there
is a well-known Sch\"{u}tte-style solution in the form of infinitary (also
called semiformal) sequent calculus with Carnap-style omega-rule that allows
full cut elimination, provably in \textbf{PA} extended by transfinite
induction up to Gentzen's ordinal $\varepsilon _{0}$ (cf. \cite{Carnap}, 
\cite{Sch}). By the same token, in the case of \textbf{PDL}, we introduce
Sch\"{u}tte-style semiformal one-sided sequent calculus \textsc{Seq}$%
_{\omega }^{\text{\textsc{pdl}}}$ whose inferences include the omega-rule
with principal formulas $\left[ P^{\ast }\right] \!A$ and prove
cut-elimination theorem using transfinite induction up to Veblen's
predicative ordinal $\varphi _{\omega }\left( 0\right) $ (that exceeds $%
\varepsilon _{0}$, see \cite{Veblen}, \cite{Fef2}). The ordinal increase in
question is caused by higher upper bounds on the complexity of cut formulas
that may contain nested occurrences of the starred programs, as compared to
plain occurrences of quantifiers in the language of \textbf{PA}. The
omega-rule-free derivations in \textsc{Seq}$_{\omega }^{\text{\textsc{pdl}}}$%
\ are finite and sequents deducible by these finite derivations are valid in 
\textbf{PDL}. Hence by the cutfree subformula property we conclude that any
given $\left[ P^{\ast }\right] $-free sequent is valid in \textbf{PDL} iff
it is deducible in \textsc{Seq}$_{\omega }^{\text{\textsc{pdl}}}$ by a
finite cut- and omega-rule free derivation, which by standard methods
enables better structural analysis of the validity of $\left[ P^{\ast }%
\right] $-free sequent involved. \footnote{{\footnotesize cf. e.g.
Gentzen-style conclusion that any given false equation }\underline{$n$}$=$%
\underline{$m$}{\footnotesize \ (in particular }$0=1${\footnotesize ) is not
valid in \textbf{PA}, since obviously it has no cutfree derivation. }} The
latter also refers to the computational complexity of the validity problem
in \textbf{PDL} that is known to be EXPTIME-complete (cf. \negthinspace \cite
{FischerLadner}, \cite{Pratt}). We show that \textbf{PDL}-validity of $\left[
P^{\ast }\right] $-free \emph{basic conjunctive normal expansions} (: $%
\mathrm{BCNE}$) is solvable in polynomial space, whereas \textbf{PDL}%
-validity of dual $\left[ P^{\ast }\right] $-free \emph{basic disjunctive
normal expansions} (: $\mathrm{BDNE}$), whose negations express that
satisfying Kripke frames encode accepting computations of polynomial-space
alternating TM, is EXPTIME-complete. Thus the conjecture \textbf{EXPTIME} = 
\textbf{PSPACE} holds true iff \textbf{PDL}-validity of $\mathrm{BDNE}$ is
decidable in polynomial space. We show that cutfree-derivability in \textsc{%
Seq}$_{\omega }^{\text{\textsc{pdl}}}$ \negthinspace (and hence \textbf{PDL}%
-validity) of any given $\mathrm{BDNE}$, $S$,\ is equivalent to the validity
of a suitable ``transparent'' quantified boolean formula $\widehat{S}$ whose
size is exponential in the size of $S$. Hence \textbf{EXPTIME = PSPACE}
holds true if $\widehat{S}$ is equivalid with another (hypothetical)
quantified boolean formula whose size is polynomial in the size of $S$, for
every $S\in $ $\mathrm{BDNE}$. This may reduce the former problem to a
complexity problem in quantified boolean logic. The whole proof is
formalized in \textbf{PA} extended by transfinite induction along $\varphi
_{\omega }\left( 0\right) $ (at most) -- actually in the corresponding
primitive recursive weakening, $\mathbf{PRA}_{\varphi _{\omega }\left(
0\right) }$.

\section{More detailed exposition}

\subsection{Hilbert-style proof system \textbf{PDL}}

\subparagraph{Language $\mathcal{L}$}

\begin{enumerate}
\item  Programs \textrm{PRO} (abbr.: $P$, $Q$, $R$, $S$, possibly indexed):

\begin{enumerate}
\item  include program-variables (\textrm{PRO}$_{0}$) $\pi _{0}$, $\pi _{1}$%
, ... (abbr.: $p$, $q$, $r$, possibly indexed),

\item  are closed under modal connectives $;$ and $\cup $ and star operation 
$^{\ast }$.
\end{enumerate}

\item  Formulas \textrm{FOR} (abbr.: $A$, $B$, $C$, $D$, $F$, $G$, $H$, $E$,
etc., possibly indexed):

\begin{enumerate}
\item  include formula-variables $\upsilon _{0}$, $\upsilon _{1}$, ...
(abbr.: $x$, $y$, $z$, possibly indexed),

\item  are closed under implication $\rightarrow $ , negation $\lnot $\ and
modal operation

$F\hookrightarrow \left[ P\right] F$, where $P\in \mathrm{PRO}$. \footnote{%
{\footnotesize Boolean constants are definable as usual e.g. by}\emph{\ }$%
1:=v_{0}\rightarrow v_{0}${\footnotesize \ and }$0:=\lnot 1${\footnotesize .}%
}
\end{enumerate}
\end{enumerate}

\textbf{Axioms (}cf. e.g. \cite{Seger}, \cite{DynamicLogic})$\mathbb{:}$ 
\footnote{{\footnotesize Standard axiom (}$D3${\footnotesize ) : }$\left[ P%
\right] \left( A\wedge B\right) \leftrightarrow \left( \left[ P\right]
A\wedge \left[ P\right] B\right) $ {\footnotesize follows by }$\left( \text{%
{\footnotesize G}}\right) ${\footnotesize \ from (}$D1${\footnotesize )--(}$%
D2${\footnotesize ), whereas (}$D6${\footnotesize )} {\footnotesize : }$%
\left[ A?\right] \!B\leftrightarrow \left( A\rightarrow B\right) $ 
{\footnotesize is obsolete\ in the chosen ?-free\ language.}}

$\left( \text{\textsc{D}}1\right) \quad Axioms\ of\ classical\
propositional\ logic.$

$\left( \text{\textsc{D}}2\right) \quad \left[ P\right] \left( A\rightarrow
B\right) \rightarrow \left( \left[ P\right] A\rightarrow \left[ P\right]
B\right) $

$\left( \text{\textsc{D}}4\right) \quad \left[ P;Q\right] A\leftrightarrow %
\left[ P\right] \!\left[ Q\right] A$

$\left( \text{\textsc{D}}5\right) \quad \left[ P\cup Q\right]
A\leftrightarrow \left[ P\right] \!A\wedge \left[ Q\right] A$

$\left( \text{\textsc{D}}7\right) \quad \left[ P^{\ast }\right]
A\leftrightarrow A\wedge \left[ P\right] \left[ P^{\ast }\right] A$

$\left( \text{\textsc{D}}8\right) \quad \left[ P^{\ast }\right] \left(
A\rightarrow \left[ P\right] A\right) \rightarrow \left( A\rightarrow \left[
P^{\ast }\right] A\right) $

\textbf{Inference rules:\medskip \qquad }

$
\begin{array}{l}
\left( \text{\textsc{MP}}\right) \quad \dfrac{A\quad A\rightarrow B}{B} \\ 
\left( \text{\textsc{G}}\right) \quad \dfrac{A}{\left[ P\right] A}
\end{array}
$

\subsection{Semiformal sequent calculus S{\protect\small EQ}$_{\protect\omega
}^{\text{PDL}}$}

\begin{definition}
The language of \textsc{Seq}$_{\omega }^{\text{\textsc{pdl}}}$ includes 
\emph{seq-formulas} and \emph{sequents}. Seq-formulas \ are built up from 
\emph{literals} $x$ and $\lnot x$ by propositional connectives $\vee $ and $%
\wedge $ and modal operations $\left[ P\right] $ and $\left\langle
P\right\rangle $ for arbitrary $P\in \mathrm{PRO}$. Seq-negation $\overline{F%
}$ is defined recursively as follows, for any seq-formula $F$.

\begin{enumerate}
\item  $\overline{x}:=\lnot x,\ \smallskip \overline{\lnot x}:=x,$

\item  $\overline{A\vee B}:=\overline{A}\wedge \overline{B},\ \overline{%
A\wedge B}:=\overline{A}\vee \overline{B}.$

\item  $\overline{\left\langle P\right\rangle \!A}:=\left[ P\right] \!%
\overline{A},\overline{\ \left[ P\right] \!A}:=\left\langle P\right\rangle 
\overline{\!A}.$
\end{enumerate}
\end{definition}

In the sequel we use abbreviations $\left\langle P\right\rangle ^{m}\!:=%
\overset{m\ times}{\overbrace{\left\langle P\right\rangle \cdots
\left\langle P\right\rangle }}$ and $\left[ P\right] ^{m}\!:=\overset{m\
times}{\overbrace{\left[ P\right] \cdots \left[ P\right] }}$. For any $\chi
\in \left\{ 0,1\right\} $, let $\!\left( P\right) _{\chi }:=\left\{ 
\begin{array}{ccc}
\left[ P\right] \text{,} & \text{if} & \chi =1\text{,} \\ 
\left\langle P\right\rangle \text{,} & \text{if} & \chi =0\text{.}
\end{array}
\right. $ For any $\overrightarrow{P}=P_{1},\cdots ,P_{k}$\ ($k\geq 0$) and $%
f:\left[ 1,k\right] \rightarrow \left\{ 0,1\right\} $ let $\left( 
\overrightarrow{P}\right) _{f}:=\left( P_{1}\right) \!_{f\left( 1\right)
}\cdots \left( P_{k}\right) \!_{f\left( k\right) }$. By $\left( 
\overrightarrow{Q}\right) $, $\left\langle \overrightarrow{Q}\right\rangle $
and $\left[ \overrightarrow{Q}\right] $ we abbreviate $\left( 
\overrightarrow{Q}\right) _{f}$\ for arbitrary $f$, $f\equiv 0$ and $f\equiv
1$, respectively. Formulas from $\mathrm{FOR}$ are represented as
seq-formulas recursively by $\lnot F:=\overline{F}$, $F\rightarrow G:=%
\overline{F}\vee G$ and, conversely, by $F\vee G:=\lnot F\rightarrow G$, $%
F\wedge G:=\lnot \left( F\rightarrow \lnot G\right) $, $\left\langle
P\right\rangle F:=\lnot \left[ P\right] \lnot F$. Sequents (abbr.: $\Gamma $%
, $\Delta $, $\Pi $, $\Sigma $, possibly indexed) are viewed as multisets
(possibly empty) of seq-formulas. A sequent $\Gamma =F_{1},\cdots ,F_{n}$ is
called \emph{valid} iff so is the corresponding disjunction $F_{1}\vee
\cdots \vee F_{n}$. \emph{Plain complexity} of a given formula and/or
program in $\mathcal{L}$ is its ordinary length (= total number of
occurrences of literals and connectives $\vee $, $\wedge \,$, $\cup $%
\thinspace , $;$\thinspace , $\ast $).

\begin{definition}
\emph{Ordinal complexity} $\frak{o\!}\left( -\right) <\omega ^{\omega }$\ of
formulas, programs and sequents in $\mathcal{L}$ is defined recursively as
follows, where $\frak{\alpha }+\!\!\!\!+\,\,\frak{\beta }$ is the symmetric
sum of ordinals $\alpha $ and $\beta $.

\begin{enumerate}
\item  $\frak{o\!}\left( x\right) =\frak{o\!}\left( \lnot x\right) =\frak{o\!%
}\left( p\right) :=0.$

\item  $\frak{o\!}\left( A\vee B\right) =\frak{o\!}\left( A\wedge B\right)
:=\max \left\{ \frak{o\!}\left( A\right) ,\frak{o\!}\left( B\right) \right\}
+1.$

\item  $\frak{o\!}\left( P\cup Q\right) :=\max \left\{ \frak{o\!}\left(
P\right) ,\frak{o\!}\left( Q\right) \right\} +1,\quad \frak{o\!}\left(
P;Q\right) :=\frak{o\!}\left( P\right) +\!\!\!\!\!+\,\,\frak{o\!}\left(
Q\right) +1.$

\item  $\frak{o\!}\left( P^{\ast }\right) :=\underset{m<\omega }{\sup }\frak{%
o\!}\left( \left( P\right) ^{m}\right) =\frak{o\!}\left( P\right) \cdot
\omega ,\quad \frak{o\!}\left( \left\langle P\right\rangle \!A\right) =\frak{%
o\!}\left( \left[ P\right] \!A\right) :=\frak{o\!}\left( P\right)
+\!\!\!\!\!+\,\,\frak{o\!}\left( A\right) +1.$

\item  $\frak{o\!}\left( \Gamma \right) :=\sum \left\{ \frak{o\!}\left(
A\right) :A\in \Gamma \right\} .$
\end{enumerate}
\end{definition}

\begin{definition}
\textsc{Seq}$_{\infty }^{\text{\textsc{pdl}}}$ includes the following axiom $%
\left( \text{\textsc{Ax}}\right) $ and inference rules $\left( \vee \right) $%
, $\left( \wedge \right) $, $\left\langle \cup \right\rangle $, $\left[ \cup %
\right] $, $\left\langle ;\right\rangle $, $\left[ ;\right] $, $\left\langle
\ast \right\rangle $, $\left[ \ast \right] $, $\left( \text{\textsc{Gen}}%
\right) $, $\left( \text{\textsc{Cut}}\right) $ in classical one-sided
sequent formalism in the language $\mathcal{L}$. In $\left[ \ast \right] $
we allow $\overrightarrow{Q}=\left[ \overrightarrow{Q}\right] =\emptyset $. 
\footnote{{\footnotesize We assume that all rules exposed have nonempty
premises. }} \footnote{$\left[ \ast \right] ${\footnotesize \ has infinitely
many premises. Ii is called the }$\omega ${\footnotesize -rule.}} 
\begin{equation*}
\begin{array}{c}
\begin{array}{c}
\fbox{$\left( \text{\textsc{Ax}}\right) \quad x,\lnot x,\Gamma $}
\end{array}
\quad \quad \qquad \\ 
\begin{array}{c}
\fbox{$\left( \vee \right) \quad \dfrac{A,B,\Gamma }{A\vee B,\Gamma }\ $}%
\fbox{$\left( \wedge \right) \quad \dfrac{A,\Gamma \quad \quad B,\Gamma }{%
A\wedge B,\Gamma }\ $}
\end{array}
\\ 
\ 
\begin{array}{c}
\fbox{$\left\langle \cup \right\rangle \quad \dfrac{\left\langle
P\right\rangle \!A,\left\langle R\right\rangle \!A,\Gamma }{\left\langle
P\cup R\right\rangle \!A,\Gamma }\ $}\fbox{$\left[ \cup \right] \quad \dfrac{%
\left[ P\right] \!A,\Gamma \text{\qquad }\left[ R\right] \!A,\Gamma }{\left[
P\cup R\right] \!A,\Gamma }\ $}
\end{array}
\, \\ 
\begin{array}{c}
\fbox{$\left\langle ;\right\rangle \quad \dfrac{\left\langle P\right\rangle
\!\!\left\langle R\right\rangle \!A,\Gamma }{\left\langle P;R\right\rangle
\!A,\Gamma }\ $}\fbox{$\left[ ;\right] \quad \dfrac{\left[ P\right] \!\left[
R\right] \!A,\Gamma }{\left[ P;R\right] \!A,\Gamma }\ $}
\end{array}
\qquad \ 
\end{array}
\end{equation*}
\begin{equation*}
\begin{array}{c}
\quad \qquad 
\begin{array}{c}
\fbox{$\left\langle \ast \right\rangle \quad \dfrac{\left\langle 
\overrightarrow{Q}\right\rangle \!\!\left\langle P\right\rangle
^{m}\!\!A,\left\langle \overrightarrow{Q}\right\rangle \!\!\left\langle
P^{\ast }\right\rangle \!A,\Gamma }{\left\langle \overrightarrow{Q}%
\right\rangle \!\!\left\langle P^{\ast }\right\rangle \!A,\Gamma }\left(
m\geq 0\right) \ $}\  \\ 
\fbox{$\left[ \ast \right] $\quad $\dfrac{\cdots \quad \left[ 
\overrightarrow{Q}\right] \!\!\left[ P\right] ^{m}\!\!A,\Gamma \quad \cdots
\ \left( \forall m\geq 0\right) }{\left[ \overrightarrow{Q}\right] \!\!\left[
P^{\ast }\right] \!A,\Gamma \quad \quad \quad \quad \quad }$}\ 
\end{array}
\\ 
\text{\ }\qquad \fbox{$
\begin{array}{c}
\left( \text{\textsc{Gen}}\right) \quad \dfrac{A_{1},\cdots ,A_{n}}{\left(
P\right) _{\chi _{1}}\!\!A_{1},\cdots ,\left( P\right) \!_{\chi
_{n}}\!A_{n},\Gamma }\ \left( n>0\right) \\ 
\text{if }\overset{n}{\underset{i=1}{\sum }}\chi _{i}=1.
\end{array}
$} \\ 
\quad \quad \quad 
\begin{array}{c}
\fbox{$\left( \text{\textsc{Cut}}\right) \quad \dfrac{C,\Gamma \quad \quad 
\overline{C},\Pi }{\Gamma \cup \Pi }$}
\end{array}
\quad
\end{array}
\end{equation*}
\end{definition}

For the sake of brevity we'll drop ``seq-'' when referring to seq-formulas
of \textsc{Seq}$_{\omega }^{\text{\textsc{pdl}}}$. $\Gamma $\ is called 
\emph{derivable} in \textsc{Seq}$_{\omega }^{\text{\textsc{pdl}}}$ if there
exists a (tree-like, possibly infinite) \textsc{Seq}$_{\omega }^{\text{%
\textsc{pdl}}}$ \emph{derivation }$\partial $\emph{\ }with the root sequent%
\emph{\ }$\Gamma $\emph{\ }(abbr.:\emph{\ }$\left( \partial :\Gamma \right) $
). We assume that \textsc{Seq}$_{\omega }^{\text{\textsc{pdl}}}$ \emph{%
derivations} are well-founded. The simplest way to implement this assumption
is to supply nodes $x$ in $\partial $ with ordinals $ord\left( x\right) $
such that ordinals of premises are always smaller that the ones of the
corresponding conclusions. Having this we let $h\left( \partial \right)
:=ord\left( root\left( \partial \right) \right) $ and call it \emph{the
height} of $\partial $.

\medskip In \textsc{Seq}$_{\omega }^{\text{\textsc{pdl}}}$, formulas
occurring in $\Gamma $ and/or $\Pi $ are called \emph{side formulas},
whereas other (distinguished) ones are called \emph{principal formulas}, of
axioms or inference rules exposed. These axioms and inferences, in turn, are
called \emph{principal} with respect to their principal formulas. Principal
formulas of $\left( \text{\textsc{Cut}}\right) $ are also called the
corresponding \emph{cut formulas}. We'll sometimes specify $\left( \text{%
\textsc{Gen}}\right) $ as $\left( \text{\textsc{Gen}}\right) _{P}$ to
indicate principal program $P$ involved.

\begin{theorem}[soundness and completeness]
\textsc{Seq}$_{\omega }^{\text{\textsc{pdl}}}$ is sound and complete with
respect to \textbf{PDL}. Moreover any \textbf{PDL}-valid sequent (in
particular formula) is derivable in \textsc{Seq}$_{\omega }^{\text{\textsc{%
pdl}}}$ using ordinals $<\frak{\omega }+\,\!\omega =:\omega \cdot 2$.
\end{theorem}

\begin{proof}
The soundness says that any sequent $\Gamma $ that is derivable in \textsc{%
Seq}$_{\omega }^{\text{\textsc{pdl}}}$\ is valid in Kripke-style semantics
of \textbf{PDL}. It is proved by transfinite induction on $h\left( \partial
\right) $ of well-founded $\left( \partial :\Gamma \right) $ involved. 
\footnote{{\footnotesize Plain (finite) induction is sufficient for }$\left[
\ast \right] ${\footnotesize -free derivations.}} Actually it suffices to
verify that every inference rule of \textsc{Seq}$_{\omega }^{\text{\textsc{%
pdl}}}$\ preserves Kripke validity, which is easy (we omit the details; see
also Remark 5 below).

The completeness is proved as follows by deducing in \textsc{Seq}$_{\omega
}^{\text{\textsc{pdl}}}$\ the axioms and inferences $\left( D1\right) $, $%
\left( D2\right) $, $\left( D4\right) $, $\left( D5\right) $, $\left(
D7\right) ,\left( D8\right) $, $\left( \text{\textsc{MP}}\right) $, $\left( 
\text{\textsc{G}}\right) $ of \textbf{PDL}. By $\overset{\mathcal{L}}{\equiv 
}$ we denote the equivalence in propositional logic.

$\left( D1\right) $ is deducible by standard method via extended axiom $%
\left( \text{\textsc{Ax}}\right) ^{+}:$ $F,\overline{F},\Gamma $ whose
finite cutfree derivation is constructed by recursion on plain complexity of 
$F$ (in particular we pass by $\left( \text{\textsc{Gen}}\right) $ from $A,%
\overline{A}$ to $\left[ P\right] \!A,\left\langle P\right\rangle \!%
\overline{A},\Gamma $).

$\left( D4\right) $ and $\left( D5\right) $ are trivial, while $\left(
D2\right) $, $\left( D7\right) $, $\left( D8\right) $ are derivable as
follows.\smallskip

$\left( D2\right) :\ \left[ P\right] \left( A\rightarrow B\right)
\rightarrow \left( \left[ P\right] A\rightarrow \left[ P\right] B\right) 
\overset{\mathcal{L}}{\equiv }\left\langle P\right\rangle \left( A\wedge 
\overline{B}\right) \vee \left\langle P\right\rangle \!\overline{A}\vee %
\left[ P\right] \!B$.

$\fbox{$\dfrac{\quad \dfrac{\quad \dfrac{\overset{\left( \text{\textsc{Ax}}%
\right) ^{+}}{A,\overline{A}},B\quad \quad \overset{\left( \text{\textsc{Ax}}%
\right) ^{+}}{\overline{B},\overline{A}},B}{A\wedge \overline{B},\overline{A}%
,B}\left( \wedge \right) }{\left\langle P\right\rangle \left( A\wedge 
\overline{B}\right) ,\left\langle P\right\rangle \!\overline{A},\left[ P%
\right] \!B}\left( \text{\textsc{Gen}}\right) }{\left\langle P\right\rangle
\left( A\wedge \overline{B}\right) \vee \left\langle P\right\rangle \!%
\overline{A}\vee \left[ P\right] \!B\quad }\left( \vee \right) $}$\medskip
\medskip

$\left( D7\right) :\ \left[ P^{\ast }\right] A\leftrightarrow A\wedge \left[
P\right] \left[ P^{\ast }\right] A$

$\overset{\mathcal{L}}{\equiv }\,\left( \left\langle P^{\ast }\right\rangle
\!\overline{A}\vee \left( A\wedge \left[ P\right] \!\left[ P^{\ast }\right]
\!A\right) \right) \wedge \left( \left[ P^{\ast }\right] \!A\vee \overline{A}%
\vee \left\langle P\right\rangle \!\left\langle P^{\ast }\right\rangle \!%
\overline{A}\right) $.

\fbox{$\dfrac{\dfrac{\dfrac{\overset{\left( \text{\textsc{Ax}}\right) ^{+}}{%
\overline{A},\left\langle P^{\ast }\right\rangle \overline{\!A},A}}{%
\left\langle P^{\ast }\right\rangle \overline{\!A},A}\left\langle \ast
\right\rangle }{\left\langle P^{\ast }\right\rangle \overline{\!A}\vee A}%
\left( \vee \right) \dfrac{\dfrac{\dfrac{\overset{\left( \text{\textsc{Ax}}%
\right) ^{+}}{\left\langle P\right\rangle ^{m+1}\!\overline{A},\left\langle
P^{\ast }\right\rangle \overline{\!A},\left[ P\right] ^{m+1}\!A}}{\cdots \
\left\langle P^{\ast }\right\rangle \!\overline{A},\left[ P\right] ^{m+1}A\
\cdots }\left\langle \ast \right\rangle }{\left\langle P^{\ast
}\right\rangle \overline{\!A},\left[ P\right] \!\left[ P^{\ast }\right] \!A}%
\left[ \ast \right] }{\left\langle P^{\ast }\right\rangle \!\overline{A}\vee %
\left[ P\right] \!\left[ P^{\ast }\right] \!A}\left( \vee \right) }{%
\left\langle P^{\ast }\right\rangle \!\overline{A}\vee \left( A\wedge \left[
P\right] \!\left[ P^{\ast }\right] \!A\right) }\left( \wedge \right) $}$\
\&\ $

\fbox{$\dfrac{\dfrac{\QDATOP{\left( \text{\textsc{Ax}}\right) ^{+}}{A,%
\overline{A},\left\langle P\right\rangle \!\!\left\langle P^{\ast
}\right\rangle \overline{\!A}}\ \ \dfrac{\overset{\left( \text{\textsc{Ax}}%
\right) ^{+}}{\left[ P\right] ^{m+1}\!A,\overline{A},\left\langle
P\right\rangle ^{m+1}\!\overline{A},\left\langle P\right\rangle
\!\left\langle P^{\ast }\right\rangle \overline{\!A}}}{\left[ P\right]
^{m+1}\!\!A,\overline{A},\left\langle P\right\rangle \!\left\langle P^{\ast
}\right\rangle \overline{\!A}\ \ \ \cdots }\left\langle \ast \right\rangle }{%
\left[ P^{\ast }\right] \!A,\overline{A},\left\langle P\right\rangle
\!\left\langle P^{\ast }\right\rangle \overline{\!A}}\left[ \ast \right] }{%
\left[ P^{\ast }\right] \!A\vee \overline{A}\vee \left\langle P\right\rangle
\!\left\langle P^{\ast }\right\rangle \overline{\!A}}\left( \vee \right) $}%
\smallskip \medskip

$\left( D8\right) :\ \left[ P^{\ast }\right] \left( A\rightarrow \left[ P%
\right] A\right) \rightarrow \left( A\rightarrow \left[ P^{\ast }\right]
A\right) $

$\overset{\mathcal{L}}{\equiv }\left\langle P^{\ast }\right\rangle \left(
A\wedge \left\langle P\right\rangle \overline{A}\right) \vee \overline{A}%
\vee \left[ P^{\ast }\right] A$.

\fbox{$\dfrac{\dfrac{\QDATOP{{}}{\cdots }\ \QDATOP{\overset{\partial _{m}}{%
\Downarrow }}{\left\langle P^{\ast }\right\rangle \!\left( A\wedge
\left\langle P\right\rangle \!\overline{A}\right) ,\overline{A},\left[ P%
\right] ^{m}\!A}\ \QDATOP{{}}{\cdots \ }\quad \quad }{\,\left\langle P^{\ast
}\right\rangle \!\left( A\wedge \left\langle P\right\rangle \!\overline{A}%
\right) ,\overline{A},\left[ P^{\ast }\right] \!A}\left[ \ast \right] }{%
\,\left\langle P^{\ast }\right\rangle \!\left( A\wedge \left\langle
P\right\rangle \!\overline{A}\right) \vee \overline{A}\vee \left[ P^{\ast }%
\right] \!A}\left( \vee \right) $} ,\smallskip\ where:\smallskip

$\partial _{0}=$\fbox{$\left\langle P^{\ast }\right\rangle \!\left( A\wedge
\left\langle P\right\rangle \!\overline{A}\right) ,\overline{A},A\mathbf{%
\quad }\left( \text{\textsc{Ax}}\right) ^{+}$}

$\partial _{1}=$\fbox{$\dfrac{\dfrac{\overset{\left( \text{\textsc{Ax}}%
\right) ^{+}}{A,\left\langle P^{\ast }\right\rangle \!\!\left( A\!\wedge
\!\left\langle P\right\rangle \!\overline{A}\right) \!,\overline{A},\left[ P%
\right] \!A}\quad \overset{\left( \text{\textsc{Ax}}\right) ^{+}}{%
\left\langle P\right\rangle \!\overline{A},\left\langle P^{\ast
}\right\rangle \!\!\left( A\!\wedge \!\left\langle P\right\rangle \!%
\overline{A}\right) \!,\overline{A},\left[ P\right] \!A}}{A\wedge
\left\langle P\right\rangle \!\overline{A},\left\langle P^{\ast
}\right\rangle \!\left( A\wedge \left\langle P\right\rangle \!\overline{A}%
\right) ,\overline{A},\left[ P\right] \!A}\left( \wedge \right) }{%
\left\langle P^{\ast }\right\rangle \!\left( A\wedge \left\langle
P\right\rangle \!\overline{A}\right) ,\overline{A},\left[ P\right] \!A}%
\left\langle \ast \right\rangle $}

$\partial _{2}=$\fbox{$\dfrac{\dfrac{\QDATOP{\left( \text{\textsc{Ax}}%
\right) ^{+}}{\,A,\Pi ^{+},\overline{A},\left[ P\right] \!A,\left[ P\right]
\!^{2}A\quad \quad }\dfrac{\dfrac{\overset{\left( \text{\textsc{Ax}}\right)
^{+}}{\,\overline{A},A,P\!A}\overset{\left( \text{\textsc{Ax}}\right) ^{+}}{%
\,\quad \quad \overline{A},\left\langle P\right\rangle \!\overline{A},\left[
P\right] \!A}}{\,\overline{A},A\wedge \left\langle P\right\rangle \!%
\overline{A},\left[ P\right] \!A}\left( \wedge \right) }{\left\langle
P\right\rangle \!\overline{A},\Pi ^{+},\overline{A},\left[ P\right] \!A,%
\left[ P\right] \!^{2}A}\left( \text{\textsc{Gen}}\right) }{A\wedge
\left\langle P\right\rangle \!\overline{A},\left\langle P\right\rangle
\!\left( A\wedge \left\langle P\right\rangle \!\overline{A}\right) ,\Pi ,%
\overline{A},\left[ P\right] \!A,\left[ P\right] \!^{2}A}\left( \wedge
\right) }{\left\langle P^{\ast }\right\rangle \!\left( A\wedge \left\langle
P\right\rangle \!\overline{A}\right) ,\overline{A},\left[ P\right] ^{2}\!A}%
\left\langle \ast \right\rangle \left\langle \ast \right\rangle $}%
\smallskip\ for $\Pi :=\left\langle P^{\ast }\right\rangle \!\left( A\wedge
\left\langle P\right\rangle \!\overline{A}\right) $ and $\Pi
^{+}:=\left\langle P\right\rangle \!\left( A\wedge \left\langle
P\right\rangle \!\overline{A}\right) ,\Pi $,

etc. via $\left( \wedge \right) $, $\left\langle \ast \right\rangle $\ and $%
\left( \text{\textsc{Gen}}\right) $.\smallskip

Obviously these derivations require ordinal assignments $<\frak{\omega }%
+\omega $. \textbf{PDL} inferences $\left( \text{\textsc{MP}}\right) $ and $%
\left( \text{\textsc{G}}\right) $ are obviously derivable by $\left( \text{%
\textsc{Cut}}\right) $ and $\left( \text{\textsc{Gen}}\right) $,
respectively. These\ increase ordinals by one, which makes an arbitrary
Hilbert-style \textbf{PDL} deduction interpretable as a \textsc{Seq}$%
_{\omega }^{\text{\textsc{pdl}}}$ derivation of the height $<\omega \cdot 2$%
, as required.
\end{proof}

\begin{remark}
The validity of $\left( \text{\textsc{Gen}}\right) $ also follows by plain
generalzation $\left( \text{\textsc{G}}\right) $ from $\left( D1\right) $, $%
\left( D2\right) $ and derivable (dual) $\left( D3\right) :\left\langle
P\right\rangle \left( A\vee B\right) \leftrightarrow \left( \left\langle
P\right\rangle A\vee \left\langle P\right\rangle B\right) $ (cf. Footnote
3), e.g. like this: 
\begin{equation*}
\fbox{$\dfrac{A_{1},A_{2},\cdots ,A_{n}\overset{\mathcal{L}}{\equiv }%
A_{1}\vee A_{2}\vee \cdots \vee A_{n}}{%
\begin{array}{c}
\left[ P\right] \!\left( A_{1}\vee A_{2}\vee \cdots \vee A_{n}\right) 
\overset{\mathcal{L}}{\equiv }\left[ P\right] \!\left( \lnot \left(
A_{2}\vee \cdots \vee A_{n}\right) \!\rightarrow \!A_{1}\right) \!\underset{%
\left( D2\right) }{\Rightarrow } \\ 
\left[ P\right] \!\left( \lnot \left( A_{2}\!\vee {}\cdots \!\vee
{}A_{n}\right) \!\rightarrow \!\left[ P\right] \!A_{1}\right) \overset{%
\mathcal{L}}{\equiv }\left[ P\right] \!A_{1}\!\vee \!\left\langle
P\right\rangle \!\left( A_{2}\!\vee \!\cdots \!\vee \!A_{n}\right) \\ 
\underset{\left( D1\right) }{\Rightarrow }\left[ P\right] \!A_{1}\vee
\left\langle P\right\rangle \!\left( A_{2}\vee \cdots \vee A_{n}\right) \vee
\Gamma \underset{\left( D3\right) }{\Rightarrow } \\ 
\left[ P\right] \!A_{1}\!\vee \!\left\langle P\right\rangle \!A_{2}\!\vee
\!\cdots \!\vee \!\left\langle P\right\rangle \!A_{n}\!\vee \!\Gamma 
\overset{\mathcal{L}}{\equiv }\left[ P\right] \!A_{1},\!\left\langle
P\right\rangle \!A_{2},\!\cdots ,\left\langle P\right\rangle \!A_{n},\Gamma
\end{array}
}\left( \text{\textsc{G}}\right) $}
\end{equation*}
\end{remark}

\subsection{Cut elimination procedure}

\subsubsection{Auxiliary sequent calculus S{\protect\small EQ}$_{\protect%
\omega \text{\textsc{+}}}^{\text{PDL}}$}

\begin{definition}
\textsc{Seq}$_{\omega \text{\textsc{+}}}^{\text{\textsc{pdl}}}$ is a
modification of \textsc{Seq}$_{\omega }^{\text{\textsc{pdl}}}$ that includes
the following upgraded inferences $\left\langle \overset{+}{\cup }%
\right\rangle $, $\left[ \overset{+}{\cup }\right] $, $\left\langle \overset{%
+}{;}\right\rangle $, $\left[ \overset{+}{;}\right] $. 
\begin{equation*}
\begin{array}{c}
\quad \qquad \, 
\begin{array}{c}
\fbox{$\left\langle \overset{+}{\cup }\right\rangle \ \dfrac{\left\langle 
\overrightarrow{Q}\right\rangle \!\!\left\langle P\right\rangle
\!A,\left\langle \overrightarrow{Q}\right\rangle \!\!\left\langle
R\right\rangle \!A,\Gamma }{\left\langle \overrightarrow{Q}\right\rangle
\!\!\left\langle P\cup R\right\rangle \!A,\Gamma }\ $}\fbox{$\left[ \overset{%
+}{\cup }\right] \ \dfrac{\left[ \overrightarrow{Q}\right] \!\!\left[ P%
\right] \!A,\Gamma \text{\qquad }\left[ \overrightarrow{Q}\right] \!\!\left[
R\right] \!A,\Gamma }{\left[ \overrightarrow{Q}\right] \!\!\left[ P\cup R%
\right] \!A,\Gamma }\ $}
\end{array}
\\ 
\,\quad \quad \quad \quad \ \quad 
\begin{array}{c}
\fbox{$\left\langle \,\overset{+}{;}\right\rangle \quad \dfrac{\left\langle 
\overrightarrow{Q}\right\rangle \!\!\left\langle P\right\rangle
\!\!\left\langle R\right\rangle \!A,\Gamma }{\left\langle \overrightarrow{Q}%
\right\rangle \!\!\left\langle P;R\right\rangle \!A,\Gamma }\ $}\fbox{$\left[
\,\overset{+}{;}\right] \quad \dfrac{\left[ \overrightarrow{Q}\right] \!\!%
\left[ P\right] \!\left[ R\right] \!A,\Gamma }{\left[ \overrightarrow{Q}%
\right] \!\!\left[ P;R\right] \!A,\Gamma }\ $}
\end{array}
\qquad \ \ \qquad \qquad \qquad
\end{array}
\end{equation*}
\end{definition}

Obviously these upgrades are still sound in \textbf{PDL} and cut-free
derivable in \textsc{Seq}$_{\omega }^{\text{\textsc{pdl}}}$. Hence \textsc{%
Seq}$_{\omega }^{\text{\textsc{pdl}}}$ and \textsc{Seq}$_{\omega \text{%
\textsc{+}}}^{\text{\textsc{pdl}}}$ are proof theoretically equivalent. In
the sequel for the sake of brevity we use old names $\left\langle \cup
\right\rangle $, $\left[ \cup \right] $, $\left\langle ;\right\rangle $, $%
\left[ ;\right] $ also for the corresponding upgrades.

\subsubsection{Admissible refinements}

\begin{lemma}
The following inferences are admissible in \textsc{Seq}$_{\omega \text{%
\textsc{+}}}^{\text{\textsc{pdl}}}$ \emph{minus}$\ \left( \text{\textsc{Cut}}%
\right) $. Moreover, for any inversion $\dfrac{\left( \partial :\Delta
\right) }{\!\left( \partial ^{\circlearrowleft }:\Gamma \right) }$ involved
we have $h\left( \partial ^{\circlearrowleft }\right) <h\left( \partial
\right) +\,\!\omega $. In $\left( \overrightarrow{\text{\textsc{Gen}}}%
\right) $ we assume that $\overrightarrow{P}=P_{1},\cdots ,P_{k}$ $\left(
k>0\right) $, $f_{1},\cdots ,f_{n}:\left[ 1,k\right] \rightarrow \left\{
0,1\right\} $ and $\left( \forall j\in \left[ 1,k\right] \right) \overset{n}{%
\underset{i=1}{\sum }}f_{i}\left( j\right) =1$. Note that $\left( \text{%
\textsc{Gen}}\right) $ is a particular case of $\left( \overrightarrow{\text{%
\textsc{Gen}}}\right) $.
\end{lemma}

\begin{eqnarray*}
&& 
\begin{array}{c}
\fbox{$\left( \text{\textsc{W}}\right) \quad \dfrac{\Gamma }{\Gamma ,\Pi }\ $%
\ (weakening)}
\end{array}
\begin{array}{c}
\fbox{$\left( \text{\textsc{C}}\right) \quad \dfrac{A,A,\Gamma }{A,\Gamma }\ 
$(contraction)}
\end{array}
\\
&& 
\begin{array}{c}
\fbox{$\left( \vee \right) ^{\circlearrowleft }\quad \dfrac{\!A\vee B,\Gamma 
}{\!A,\!B,\Gamma }\ $}
\end{array}
\begin{array}{c}
\fbox{$\left( \wedge \right) _{1}^{\circlearrowleft }\quad \dfrac{\!A\wedge
B,\Gamma }{\!A,\Gamma }\ $}
\end{array}
\begin{array}{c}
\fbox{$\left( \wedge \right) _{2}^{\circlearrowleft }\quad \dfrac{\!A\wedge
B,\Gamma }{\!B,\Gamma }\ $}
\end{array}
\\
&& 
\begin{array}{c}
\fbox{$\left\langle \cup \right\rangle ^{\circlearrowleft }\quad \dfrac{%
\!\left\langle \overrightarrow{Q}\right\rangle \!\!\left\langle P\cup
R\right\rangle \!A,\Gamma }{\!\left\langle \overrightarrow{Q}\right\rangle
\!\!\left\langle P\right\rangle \!A,\!\left\langle \overrightarrow{Q}%
\right\rangle \!\!\left\langle R\right\rangle \!A,\Gamma }\ $}
\end{array}
\\
&& 
\begin{array}{c}
\fbox{$\left[ \cup \right] _{1}^{\circlearrowleft }\quad \dfrac{\!\left[ 
\overrightarrow{Q}\right] \!\!\left[ P\cup R\right] \!A,\Gamma }{\left[ 
\overrightarrow{Q}\right] \!\!\left[ P\right] \!A,\Gamma }\ $}
\end{array}
\begin{array}{c}
\fbox{$\left[ \cup \right] _{2}^{\circlearrowleft }\quad \dfrac{\left[ 
\overrightarrow{Q}\right] \!\!\left[ P\cup R\right] \!A,\Gamma }{\!\left[ 
\overrightarrow{Q}\right] \!\!\left[ R\right] \!A,\Gamma }\ $}
\end{array}
\\
&& 
\begin{array}{c}
\fbox{$\left\langle \,\cup \right\rangle ^{\circlearrowleft }\quad \dfrac{%
\!\left\langle \overrightarrow{Q}\right\rangle \!\!\left\langle
P;\!R\right\rangle \!A,\Gamma }{\!\left\langle \overrightarrow{Q}%
\right\rangle \!\!\left\langle P\right\rangle \!\!\left\langle
R\right\rangle \!A,\Gamma }\ $}
\end{array}
\begin{array}{c}
\fbox{$\,\left[ \cup \right] ^{\circlearrowleft }\quad \dfrac{\left[ 
\overrightarrow{Q}\right] \!\!\left[ P;\!R\right] \!A,\Gamma }{\left[ 
\overrightarrow{Q}\right] \!\!\left[ P\right] \!\left[ R\right] \!A\!,\Gamma 
}\ $}
\end{array}
\end{eqnarray*}
\begin{eqnarray*}
&& 
\begin{array}{c}
\fbox{$\left[ \ast \right] ^{\circlearrowleft }\quad \dfrac{\left[ 
\overrightarrow{Q}\right] \!\!\left[ P^{\ast }\right] \!A,\Gamma }{\left[ 
\overrightarrow{Q}\right] \!\!\left[ P\right] ^{m}\!A,\Gamma }\ \left( m\geq
0\right) \ $}
\end{array}
\\
&& 
\begin{array}{c}
\fbox{$\left( \overrightarrow{\text{\textsc{Gen}}}\right) \quad \dfrac{%
A_{1},\cdots ,A_{n}}{\left( \overrightarrow{P}\right) _{\text{\negthinspace }%
f_{1}}\!\!\!A_{1},\cdots ,\left( \overrightarrow{P}\right)
_{\!f_{n}}\!\!\!A_{n},\Gamma }\ \left( n>0\right) $}
\end{array}
\end{eqnarray*}

\begin{proof}
Induction on proof height and/or formula complexity. Cases $\left( \text{%
\textsc{W}}\right) $, $\left( \text{\textsc{C}}\right) $ are standard. Note
that $\left( \text{\textsc{C}}\right) $ with principal $\left( \text{\textsc{%
Gen}}\right) $ is trivial, e.g.

$\partial :\ $\fbox{$\dfrac{\left( \partial _{1}:A,A,B\right) }{\left\langle
P\right\rangle \!A,\left\langle P\right\rangle \!A,\left[ P\right] B,\Gamma }%
\ \left( \text{\textsc{Gen}}\right) $}$\quad \hookrightarrow \quad \partial
^{\text{\textsc{C}}}:\ $\fbox{$\dfrac{\left( \partial _{1}^{\text{\textsc{C}}%
}:A,B\right) }{\left\langle P\right\rangle \!A,\left[ P\right] B,\Gamma }\
\left( \text{\textsc{Gen}}\right) $}.\smallskip

\textbf{Case} $\left( \overrightarrow{\text{\textsc{Gen}}}\right) $ is an
obvious iteration of $\left( \text{\textsc{Gen}}\right) $.

\textbf{Cases }$\left( \vee \right) ^{\circlearrowleft }$, $\left( \wedge
\right) _{1}^{\circlearrowleft }$, $\left( \wedge \right)
_{2}^{\circlearrowleft }$ are standard (and trivial) boolean inversions%
\textbf{.}

\textbf{Case }$\left\langle \overset{+}{\cup }\right\rangle
^{\circlearrowleft }$\textbf{\ }(\textbf{\ }$\left[ \overset{+}{\cup }\right]
^{\circlearrowleft }$ analogous)\textbf{. }We omit trivial case of principal
inversion of $\left\langle \cup \right\rangle $ and show only the crucial
cases of principal $\left( \text{\textsc{Gen}}\right) $ (in simple
form):\smallskip

$\partial :\ $\fbox{$\dfrac{\left( \partial _{1}:A,B,C\right) }{\left\langle
P\cup R\right\rangle \!A,\left\langle P\cup R\right\rangle \!B,\left[ P\cup R%
\right] \!C,\Gamma }\ \left( \text{\textsc{Gen}}\right) $}$\quad
\hookrightarrow \partial ^{\left\langle \cup \right\rangle
^{\circlearrowleft }}:\smallskip $

\fbox{$\dfrac{\dfrac{\left( \text{\textsc{Gen}}\right) \!_{\!P}\dfrac{\left(
\partial _{1}:A,B,C\right) }{\left\langle P\right\rangle
\!\!A,\!\left\langle R\right\rangle \!\!A,\!\left\langle P\right\rangle
\!\!B,\!\left\langle R\right\rangle \!\!B,\!\left[ P\right] \!C,\!\Gamma }}{%
\left\langle P\right\rangle \!\!A,\!\left\langle R\right\rangle
\!\!A,\!\left\langle P\!\cup \!\!R\right\rangle \!\!B,\!\left[ P\right]
\!C,\!\Gamma }\!\left\langle \cup \right\rangle \!\dfrac{\dfrac{\left(
\partial _{1}:A,B,C\right) }{\left\langle P\right\rangle
\!\!A,\!\left\langle R\right\rangle \!\!A,\!\left\langle P\right\rangle
\!\!B,\!\left\langle R\right\rangle \!\!B,\!\left[ R\right] \!C,\!\Gamma }%
\!\left( \text{\textsc{Gen}}\right) \!_{\!R}}{\left\langle P\right\rangle
\!\!A,\!\left\langle R\right\rangle \!\!A,\!\left\langle P\!\cup
\!\!R\right\rangle \!\!B,\!\left[ R\right] \!C,\!\Gamma }}{\left\langle
P\right\rangle \!\!A,\!\left\langle R\right\rangle \!\!A,\!\left\langle
P\!\cup \!\!R\right\rangle \!\!B,\!\left[ P\!\cup \!\!R\right] \!C,\!\Gamma }%
\!\left[ \cup \right] $},\smallskip

$\partial :\ $\fbox{$\dfrac{\left( \partial _{1}:\left\langle P\cup
R\right\rangle \!A,B\right) }{\left\langle Q\right\rangle \!\!\left\langle
P\cup R\right\rangle \!A,\left[ Q\right] \!B,\Gamma }\left( \text{\textsc{Gen%
}}\right) $}$\quad \hookrightarrow $

$\partial ^{\left\langle \cup \right\rangle ^{\circlearrowleft }}:\ $\fbox{$%
\dfrac{\left( \partial _{1}^{\left\langle \cup \right\rangle
^{\circlearrowleft }}:\left\langle P\right\rangle \!A,\left\langle
R\right\rangle \!A,B\right) }{\left\langle Q\right\rangle \!\!\left\langle
P\right\rangle \!A,\left\langle Q\right\rangle \!\!\left\langle
R\right\rangle \!A,\left[ Q\right] \!B,\Gamma }\left( \text{\textsc{Gen}}%
\right) $}.\smallskip

\textbf{Case }$\left\langle \,;\right\rangle ^{\circlearrowleft }$\textbf{\ }%
($\left[ \,;\,\right] ^{\circlearrowleft }$ analogous)\textbf{. }As above,
we omit trivial case of principal inversion of $\left\langle
\,;\right\rangle $ and show the crucial cases of principal $\left( \text{%
\textsc{Gen}}\right) $\ (in simple form):

$\partial :\ $\fbox{$\dfrac{\left( \partial _{1}:A,B,C\right) }{\left\langle
P;\!R\right\rangle \!A,\left\langle P;\!R\right\rangle \!B,\left[ P;\!R%
\right] \!C,\Gamma }\ \left( \text{\textsc{Gen}}\right) $}$\quad
\hookrightarrow $

$\partial ^{\left\langle ;\right\rangle ^{\circlearrowleft }}:\ $\fbox{$%
\dfrac{\dfrac{\dfrac{\left( \partial _{1}:A,B,C\right) }{\left\langle
P\right\rangle \!\!\left\langle R\right\rangle \!A,\left\langle
P\right\rangle \!\!\left\langle R\right\rangle \!B,\left[ P\right] \!\left[ R%
\right] \!C,\Gamma }\left( \overrightarrow{\text{\textsc{Gen}}}\right) _{R,P}%
}{\left\langle P\right\rangle \!\!\left\langle R\right\rangle
\!A,\left\langle P;\!R\right\rangle \!B,\left[ P\right] \!\left[ R\right]
\!C,\Gamma }\left\langle ;\right\rangle }{\left\langle P\right\rangle
\!\!\left\langle R\right\rangle \!A,\left\langle P;\!R\right\rangle \!B,%
\left[ P;\!R\right] \!C,\Gamma }\ \left[ ;\right] $},$\smallskip $

$\partial :\ $\fbox{$\dfrac{\left( \partial _{1}:\left\langle
P;\!R\right\rangle \!A,B,C\right) }{\left\langle Q\right\rangle
\!\!\left\langle P;\!R\right\rangle \!A,\left\langle Q\right\rangle \!B,%
\left[ Q\right] \!C,\Gamma }\ \left( \text{\textsc{Gen}}\right) $}$\quad
\hookrightarrow $

$\partial ^{\left\langle ;\right\rangle ^{\circlearrowleft }}:\ $\fbox{$%
\dfrac{\left( \partial _{1}^{\left\langle ;\right\rangle ^{\circlearrowleft
}}:\left\langle P\right\rangle \!\!\left\langle R\right\rangle
\!A,B,C\right) }{\left\langle Q\right\rangle \!\!\left\langle P\right\rangle
\!\!\left\langle R\right\rangle \!A,\left\langle Q\right\rangle \!B,\left[ Q%
\right] \!C,\Gamma }\left( \text{\textsc{Gen}}\right) $}$.\smallskip $

\textbf{Case }$\left[ \ast \right] ^{\circlearrowleft }$ is analogous to $%
\left( \wedge \right) _{i}^{\circlearrowleft }$, via trivial inversion of $%
\left[ \ast \right] $:\smallskip

$\partial :\ \fbox{$\dfrac{\left( \partial _{1}:A,B\right) }{\left[ P^{\ast }%
\right] \!A,\left\langle P^{\ast }\right\rangle \!B,\Gamma }\ \left( \text{%
\textsc{Gen}}\right) $}\quad \hookrightarrow $

$\partial ^{\left[ \ast \right] ^{\circlearrowleft }}:\ \fbox{$\dfrac{\dfrac{%
\left( \partial _{1}:A,B\right) }{\left[ P\right] ^{m}\!\!A,\left\langle
P\right\rangle ^{m}\!\!B,\left\langle P^{\ast }\right\rangle \!B,\Gamma }%
\left( \overrightarrow{\text{\textsc{Gen}}}\right) \underset{m}{_{%
\underbrace{P\cdots P}}}}{\left[ P\right] ^{m}\!\!A,\left\langle P^{\ast
}\right\rangle \!B,\Gamma }\ \left\langle \ast \right\rangle $}$,\smallskip

$\partial :\ \fbox{$\dfrac{\left( \partial _{1}:\left[ \overrightarrow{R}%
\right] \!\!\left[ P^{\ast }\right] \!A,B\right) }{\left[ Q\right] \!\!\left[
\overrightarrow{R}\right] \!\!\left[ P^{\ast }\right] \!A,\left\langle
Q\right\rangle \!B,\Gamma }\ \left( \text{\textsc{Gen}}\right) $}\quad
\hookrightarrow $

$\partial ^{\left[ \ast \right] ^{\circlearrowleft }}:\ \fbox{$\dfrac{\left(
\partial _{1}^{\left[ \ast \right] ^{\circlearrowleft }}:\left[ 
\overrightarrow{R}\right] \!\!\left[ P\right] ^{m}\!\!A,B,\Gamma \right) }{%
\left[ Q\right] \!\!\left[ \overrightarrow{R}\right] \!\!\left[ P\right]
^{m}\!\!A,\left\langle Q\right\rangle \!B,\Gamma }\ \left( \text{\textsc{Gen}%
}\right) $}$.\smallskip

Note that $\left( \text{\textsc{W}}\right) $, $\left( \text{\textsc{C}}%
\right) $, $\left( \vee \right) ^{\circlearrowleft }$, $\left( \wedge
\right) _{1}^{\circlearrowleft }$, $\left( \wedge \right)
_{2}^{\circlearrowleft }$ don't increase derivation heights.
\end{proof}

\subsubsection{Cut elimination proper}

We adapt familiar predicative cut elimination techniques (\cite{Sch}, \cite
{Fef1}, \cite{Poh}, \cite{Gordeev}, \cite{Buchh}).

\begin{theorem}[Predicative cut elimination]
The following is provable in \textbf{PA} extended by transfinite induction
up to Veblen-Feferman ordinal $\varphi _{\omega }\!\left( 0\right)
>\varepsilon _{0}$. Any sequent derivable in \textsc{Seq}$_{\omega }^{\text{%
\textsc{pdl}}}$ is derivable in \textsc{Seq}$_{\omega \text{\textsc{+}}}^{%
\text{\textsc{pdl}}}$ minus $\left( \text{\textsc{Cut}}\right) $. Hence any 
\textbf{PDL}-valid sequent (formula) is derivable in the cut-free fraction
of \textsc{Seq}$_{\omega \text{\textsc{+}}}^{\text{\textsc{pdl}}}$, and
hence also in \textsc{Seq}$_{\omega }^{\text{\textsc{pdl}}}$ minus $\left( 
\text{\textsc{Cut}}\right) $.
\end{theorem}

\begin{proof}
Our \emph{cut elimination operator} $\partial \hookrightarrow \mathcal{E}%
\left( \partial \right) $ satisfying $\deg \left( \mathcal{E}\left( \partial
\right) \right) =0$ is defined for any derivation $\partial $ in \textsc{Seq}%
$_{\omega \text{\textsc{+}}}^{\text{\textsc{pdl}}}$ by simultaneous
transfinite recursion on $h\left( \partial \right) $ and ordinal cut-degree $%
\deg \left( \partial \right) $. 
\begin{equation*}
\fbox{$\deg \left( \partial \right) :=\max \left\{ 0,\sup \left\{ \frak{o\!}%
\left( C\right) +1:C\ occurs\ as\ cut\ formula\ in\ \partial \right\}
\right\} $}
\end{equation*}
Namely, for any inference rule $\left( \text{\textsc{R}}\right) \neq \left( 
\text{\textsc{Cut}}\right) $ with 
\begin{eqnarray*}
&&\fbox{$\left( \partial :\Gamma \right) =\dfrac{\left( \partial _{1}:\Gamma
_{1}\right) }{\Gamma }\ \left( \text{\textsc{R}}\right) $}\text{ \ , \ }%
\fbox{$\left( \partial :\Gamma \right) =\dfrac{\left( \partial _{1}:\Gamma
_{1}\right) \quad \quad \left( \partial _{2}:\Gamma _{2}\right) }{\Gamma }\
\left( \text{\textsc{R}}\right) $}\text{ } \\
&&\text{\ \qquad \qquad\ or \ }\fbox{$\left( \partial :\Gamma \right) =%
\dfrac{\cdots \quad \left( \partial _{m}:\Gamma _{m}\right) \quad \cdots
\left\{ m\geq 0\right\} }{\Gamma }\ \left( \text{\textsc{R}}\right) =\left[
\ast \right] $}
\end{eqnarray*}
we respectively let 
\begin{eqnarray*}
&&\fbox{$\left( \mathcal{E\!}\left( \partial \right) \!:\!\Gamma \right)
\!=\!\dfrac{\left( \mathcal{E\!}\left( \partial _{1}\right) :\Gamma
_{1}\right) }{\Gamma }\left( \text{\textsc{R}}\right) $}\text{ \ or \ }\fbox{%
$\left( \mathcal{E\!}\left( \partial \right) \!:\!\Gamma \right) \!=\!\dfrac{%
\left( \mathcal{E\!}\left( \partial _{1}\right) :\Gamma _{1}\right) \ \
\left( \mathcal{E\!}\left( \partial _{2}\right) :\Gamma _{2}\right) }{\Gamma 
}\left( \text{\textsc{R}}\right) $} \\
&&\text{\ \qquad\ \qquad or \ }\fbox{$\left( \mathcal{E\!}\left( \partial
\right) :\Gamma \right) =\dfrac{\cdots \quad \left( \mathcal{E\!}\left(
\partial _{m}\right) :\Gamma _{m}\right) \quad \cdots \left\{ m\geq
0\right\} }{\Gamma }\ \left[ \ast \right] $}\text{.}
\end{eqnarray*}
Otherwise, if $\left( \text{\textsc{R}}\right) =\left( \text{\textsc{Cut}}%
\right) $ with 
\begin{equation*}
\fbox{$\left( \partial :\Gamma \cup \Pi \right) =\dfrac{\left( \partial
_{1}:C,\Gamma \right) \quad \quad \left( \partial _{2}:\overline{C},\Pi
\right) }{\Gamma \cup \Pi }\ \left( \text{\textsc{Cut}}\right) $}
\end{equation*}
then we stipulate 
\begin{equation*}
\fbox{$\left( \mathcal{E\!}\left( \partial \right) :\Gamma \cup \Pi \right)
=\left( \mathcal{E\!}\left( \!\mathcal{R\!}\left( \dfrac{\left( \mathcal{E\!}%
\left( \partial _{1}\right) :C,\Gamma \right) \quad \left( \mathcal{E\!}%
\left( \partial _{2}\right) :\overline{C},\Pi \right) }{\Gamma \cup \Pi }%
\left( \text{\textsc{Cut}}\right) \!\right) \!\right) \!:\Gamma \cup \Pi
\!\right) $}
\end{equation*}
with respect to a suitable \emph{cut reduction operation} $\partial
\hookrightarrow \mathcal{R}\left( \partial \right) $ such that 
\begin{equation*}
\fbox{$\deg \left( \mathcal{R}\left( \partial \right) \right) <\deg \left(
\partial \right) $ if $\deg \left( \partial _{1}\right) =\deg \left(
\partial _{2}\right) =0$}\text{,}
\end{equation*}
which makes $\mathcal{E\!}\left( \partial \right) $, $\deg \left( \mathcal{E}%
\left( \partial \right) \right) =0$, definable by induction on $\deg \left(
\partial \right) $ and $h\left( \partial \right) $.

Now $\mathcal{R}\left( \partial \right) $ is defined for any 
\begin{equation*}
\fbox{$\left( \partial :\Gamma \cup \Pi \right) =\dfrac{\left( \partial
_{1}:C,\Gamma \right) \quad \quad \left( \partial _{2}:\overline{C},\Pi
\right) }{\Gamma \cup \Pi }\ \left( \text{\textsc{Cut}}\right) $}
\end{equation*}
by following double induction on ordinal complexity of $C$ and $\max \left(
h\left( \partial _{1}\right) ,h\left( \partial _{2}\right) \right) $,
provided that $\deg \left( \partial _{1}\right) =\deg \left( \partial
_{2}\right) =0$.

\textbf{1.} Case $C=L$ and $\overline{C}=\overline{L}$ for $L\in \left\{
x,\lnot x\right\} $. This case is standard. Namely, we observe that $L$ is
principal left-hand side cut formula only if $\left( \partial _{1}:L,\Gamma
\right) $ with $\Gamma =\overline{L},\Gamma ^{\prime }$. But then $\left(
\partial _{2}:\overline{L},\Pi \right) $ infers $\Gamma \cup \Pi =$ $%
\overline{L},\Gamma ^{\prime },\Pi $ by derivable weakening $\left( \text{%
\textsc{W}}\right) $. Thus following Mints-style graphic presentation, we
can construct $\mathcal{R}\left( \partial \right) $ by (1) setting $L:=\Pi $
for every non-principal predecessor of the left-hand cut formula $L$ while
ascending $\partial _{1}:L,\Gamma $ either up to its disappearance as a side
formula of $\left( \text{\textsc{Gen}}\right) $ -- then we are done -- or
else up to its principal occurrence in $\left( \text{\textsc{Ax}}\right) $ $%
\overline{L},L,\Gamma ^{\prime \prime }$, which yields sequent $\overline{L}%
,\Pi ,\Gamma ^{\prime \prime }$ instead, followed in the latter case by (2)
setting $\overline{L}:=\overline{L},\Gamma ^{\prime \prime }$ for every
non-principal predecessor of the right-hand cut formula $\overline{L}$ while
ascending $\partial _{2}:$ $\overline{L},\Pi $ up to its disappearance as a
side formula of $\left( \text{\textsc{Gen}}\right) $ or else up to any
occurrence in a leaf. In both cases we eventually arrive at a correct
derivation of $\Gamma \cup \Pi $, since in the case of principal occurrence
of $\overline{L}$ in a modified axiom $\left( \text{\textsc{Ax}}\right) $ $%
\overline{L},L,\Pi ^{\prime }$ we obtain another axiom $\left( \text{\textsc{%
Ax}}\right) $ $\overline{L},\Gamma ^{\prime \prime },L,\Pi ^{\prime }$.

\textbf{2.} Case $C=A\vee B$ and $\overline{C}=\overline{A}\wedge B$. Use
derivable inversions $\left( \vee \right) ^{\circlearrowleft },\left( \wedge
\right) _{1}^{\circlearrowleft },\left( \wedge \right)
_{2}^{\circlearrowleft }$:

$\partial :\ $\fbox{$\dfrac{\left( \partial _{1}:A\vee B,\Gamma \right)
\quad \left( \partial _{2}:\overline{A}\wedge \overline{B},\Pi \right) }{%
\Gamma \cup \Pi }\left( \text{\textsc{Cut}}\right) $}$\quad \hookrightarrow
\quad $

$\mathcal{R}\left( \partial \right) :=\ $\fbox{$\dfrac{\dfrac{\left(
\partial _{1}^{\left( \vee \right) ^{\circlearrowleft }}:A,B,\Gamma \right)
\quad \left( \partial _{2}^{\left( \vee \right) _{1}^{\circlearrowleft }}:%
\overline{A},\Pi \right) }{B,\Gamma \cup \Pi }\left( \text{\textsc{Cut}}%
\right) \QDATOP{{}}{\left( \partial _{2}^{\left( \vee \right)
_{2}^{\circlearrowleft }}:\overline{B},\Pi \right) }}{\Gamma \cup \Pi }%
\left( \text{\textsc{Cut}}\right) $}.\smallskip

\textbf{3.} Case $C=\!\left\langle \overrightarrow{Q}\right\rangle
\!\!\left\langle P\cup R\right\rangle \!A$ and $\overline{C}=\!\left[ 
\overrightarrow{Q}\right] \!\!\left[ P\cup R\right] \overline{A}$. Analogous
reduction to $\left( \text{\textsc{Cut}}\right) $'s on $\left\langle 
\overrightarrow{Q}\right\rangle \!\!\left\langle P\right\rangle \!A$ and $%
\left\langle \overrightarrow{Q}\right\rangle \!\!\left\langle R\right\rangle
\!A$ by derivable inversions $\left\langle \cup \right\rangle
^{\circlearrowleft }$, $\left[ \cup \right] _{1}^{\circlearrowleft }$, $%
\left[ \cup \right] _{2}^{\circlearrowleft }$.

\textbf{4.} Case $C=\!\left\langle \overrightarrow{Q}\right\rangle
\!\!\left\langle P;R\right\rangle \!A$ and $\overline{C}=\!\left[ 
\overrightarrow{Q}\right] \!\!\left[ P;R\right] \overline{A}$. Immediate
reduction to $\left( \text{\textsc{Cut}}\right) $ on $\left\langle 
\overrightarrow{Q}\right\rangle \!\!\left\langle P\right\rangle
\!\!\left\langle R\right\rangle \!A$ by derivable inversions $\left\langle
\,;\right\rangle ^{\circlearrowleft }$, $\left[ \,;\right]
^{\circlearrowleft }$.

\textbf{5.} Case $C=\left\langle \overrightarrow{Q}\right\rangle \!F$ and $%
\overline{C}=\left[ \overrightarrow{Q}\right] \overline{F}$ where $%
\overrightarrow{Q}=Q_{1}\cdots Q_{n}$ $\left( n>0\right) $ and $\left(
\forall j\in \left[ 1,n\right] \right) \left( Q_{j}\!=p_{j}\text{ or }%
Q_{j}=P_{j}^{\ast }\right) $, while $F\neq \left\langle Q\right\rangle
F^{\prime }$. The reduction is either trivial, if $\partial _{1}=\left( 
\text{\textsc{Ax}}\right) ^{+}$, or\smallskip\ else defined hereditarily
with respect to left-hand side non-principal subcases like\smallskip

$\partial _{1}:\ \fbox{$\dfrac{\left( \partial _{1}^{\prime }:C,\Gamma
^{\prime }\right) }{C,\Gamma }\ \left( \text{\textsc{R}}\right) $}$%
\smallskip\ with$\ \partial _{2}:\left[ p\right] \!\overline{A},\Pi $, when
we let

$\mathcal{R}\left( \partial \right) :=\fbox{$\dfrac{\left( \mathcal{R}\left(
\partial ^{\prime }\right) :\Gamma ^{\prime }\cup \Pi \right) }{\qquad
\Gamma \cup \Pi \quad \quad }\left( \text{\textsc{R}}\right) $}\smallskip $
\ for $\partial ^{\prime }:\ $\fbox{$\dfrac{\left( \partial _{1}^{\prime
}:C,\Gamma ^{\prime }\right) \quad \quad \left( \partial _{2}:\overline{C}%
,\Pi \right) }{\Gamma ^{\prime }\cup \Pi }\left( \text{\textsc{Cut}}\right) $%
},\smallskip

or analogous non-principal subcases $\partial _{1}:\ \fbox{$\dfrac{\left(
\partial _{1}^{\prime }:C,\Gamma ^{\prime }\right) \quad \left( \partial
_{1}^{\prime \prime }:C,\Gamma ^{\prime \prime }\right) }{C,\Gamma }\ \left( 
\text{\textsc{R}}\right) $}\smallskip $,

$\partial _{1}:\ $\ $\fbox{$\dfrac{\cdots \quad \left( \partial ^{\left(
m\right) }:C,\Gamma ^{\left( m\right) }\right) \quad \cdots \left( \forall
m\geq 0\right) }{C,\Gamma }\ \left[ \ast \right] $}$,\smallskip

as well as the following principal subcases 5 $\left( \text{a}\right) $, 5 $%
\left( \text{b}\right) $, 5 $\left( \text{c}\right) $.\smallskip

\textbf{5 }$\left( \text{\textbf{a}}\right) $.\textbf{\smallskip } $%
C=\left\langle \overrightarrow{Q}\right\rangle \!F=\left\langle 
\overrightarrow{Q^{\prime }}\right\rangle \!\!\left\langle P^{\ast
}\right\rangle \!\!A$ and

$\partial _{1}:\ $\fbox{$\dfrac{\left( \partial _{1}^{\prime }:\left\langle 
\overrightarrow{Q^{\prime }}\right\rangle \!\!\left\langle P\right\rangle
\!^{m}A,\left\langle \overrightarrow{Q^{\prime }}\right\rangle
\!\!\left\langle P^{\ast }\right\rangle \!\!A,\Gamma \right) }{\left\langle 
\overrightarrow{Q^{\prime }}\right\rangle \!\!\left\langle P^{\ast
}\right\rangle \!\!A,\Gamma }\left\langle \ast \right\rangle $}\smallskip\
with\ $\partial _{2}:\left[ \overrightarrow{Q^{\prime }}\right] \!\!\left[
P^{\ast }\right] \!\overline{A},\Pi $. Let

$\mathcal{R}\left( \partial \right) :=$\smallskip

\fbox{$\dfrac{\left( \mathcal{R}\left( \partial ^{\prime }\right)
:\left\langle \overrightarrow{Q^{\prime }}\right\rangle \!\!\left\langle
P\right\rangle \!^{m}A,\Gamma \cup \Pi \right) \quad \quad \left( \partial
_{2}^{\left[ \ast \right] ^{\circlearrowright }}:\left[ \overrightarrow{%
Q^{\prime }}\right] \!\!\left[ P\right] \!^{m}\overline{A},\Pi \right) }{%
\Gamma \cup \Pi }\left( \text{\textsc{Cut}}\right) $} \ where$\smallskip $

$\partial ^{\prime }:\ $\fbox{$\dfrac{\left( \partial _{1}^{\prime
}:\left\langle \overrightarrow{Q^{\prime }}\right\rangle \!\!\left\langle
P\right\rangle \!^{m}A,\left\langle \overrightarrow{Q^{\prime }}%
\right\rangle \!\!\left\langle P^{\ast }\right\rangle \!\!A,\Gamma \right)
\quad \left( \partial _{2}:\left[ \overrightarrow{Q^{\prime }}\right] \!\!%
\left[ P^{\ast }\right] \!\overline{A},\Pi \right) }{\left\langle 
\overrightarrow{Q^{\prime }}\right\rangle \!\!\left\langle P\right\rangle
\!^{m}A,\Gamma \cup \Pi }\left( \text{\textsc{Cut}}\right) $}.\smallskip

\textbf{5 }$\left( \text{\textbf{b}}\right) $.\smallskip\ $C=\left\langle 
\overrightarrow{Q}\right\rangle \!F=\left\langle P^{\ast }\right\rangle \!A$%
\ and

$\partial _{1}:\ $\fbox{$\dfrac{\left( \partial _{1}^{\prime }:A,%
\overrightarrow{B},D\right) }{%
\begin{array}{c}
\left\langle P^{\ast }\right\rangle \!A,\left\langle P^{\ast }\right\rangle
\!\overrightarrow{B},\left[ P^{\ast }\right] \!D,\Gamma ^{\prime } \\ 
=\left\langle P^{\ast }\right\rangle \!A,\Gamma
\end{array}
}\left( \text{\textsc{Gen}}\right) $}\smallskip\ with$\ \partial _{2}:\left[
P^{\ast }\right] \!\overline{A},\Pi $. Then let

$\mathcal{R}\left( \partial \right) :=\ $\fbox{$\dfrac{\cdots \ \left(
\partial _{m}^{\prime \prime }:\left\langle P^{\ast }\right\rangle \!%
\overrightarrow{B},\left[ P\right] \!^{m}D,\Gamma ^{\prime }\cup \Pi \right)
\ \cdots \left( \forall m\geq 0\right) }{%
\begin{array}{c}
\left\langle P^{\ast }\right\rangle \!\overrightarrow{B},\left[ P^{\ast }%
\right] \!D,\Gamma ^{\prime }\cup \Pi \\ 
=\Gamma \cup \Pi
\end{array}
}\left[ \ast \right] $}\smallskip\ \ where

$\partial _{m}^{\prime \prime }:=\smallskip $

\fbox{$\dfrac{\dfrac{\dfrac{\left( \partial _{1}^{\prime }:A,\overrightarrow{%
B},D\right) }{\left\langle P\right\rangle \!^{m}\!A,\left\langle
P\right\rangle \!^{m}\!\overrightarrow{B},\left[ P\right] \!^{m}\!D,\left%
\langle P^{\ast }\right\rangle \!\!\overrightarrow{B},\Gamma ^{\prime }}%
\!\left( \overrightarrow{\text{\textsc{Gen}}}\right) \ \QDATOP{{}}{\left(
\partial _{2}^{\left[ \ast \right] ^{\circlearrowright }}:\left[ P\right]
\!^{m}\overline{A},\Pi \right) }}{\left\langle P\right\rangle \!^{m}\!%
\overrightarrow{B},\left\langle P^{\ast }\right\rangle \!\!\overrightarrow{B}%
,\left[ P\right] \!^{m}\!D,\Gamma ^{\prime }\cup \Pi }\!\left( \text{\textsc{%
Cut}}\right) }{\left\langle P^{\ast }\right\rangle \!\!\overrightarrow{B},%
\left[ P\right] \!^{m}\!D,\Gamma ^{\prime }\cup \Pi }\!\left\langle \ast
\right\rangle $}.\smallskip

\textbf{5 }$\left( \text{\textbf{c}}\right) $.\textbf{\smallskip } $%
C=\left\langle \overrightarrow{Q}\right\rangle \!F=\left\langle
p\right\rangle \!A$ and

$\partial _{1}:\ $\fbox{$\dfrac{\left( \partial _{1}^{\prime }:A,%
\overrightarrow{B},D\right) }{%
\begin{array}{c}
\left\langle p\right\rangle \!A,\left\langle p\right\rangle \!%
\overrightarrow{B},\left[ p\right] \!D,\Gamma ^{\prime } \\ 
=\left\langle p\right\rangle \!A,\Gamma
\end{array}
}\left( \text{\textsc{Gen}}\right) $}\smallskip\ with$\ \partial _{2}:\left[
p\right] \!\overline{A},\Pi $. Then we let

$\mathcal{R}\left( \partial \right) :=\ $\fbox{$\dfrac{\left( \mathcal{R}%
^{\prime }\left( \partial _{1}^{\prime },\partial _{2}\right) :\left\langle
p\right\rangle \!\overrightarrow{B},\left[ p\right] \!D,\Pi \right) }{%
\begin{array}{c}
\left\langle p\right\rangle \!\overrightarrow{B},\left[ p\right] \!D,\Gamma
^{\prime }\cup \Pi \\ 
=\Gamma \cup \Pi
\end{array}
}\left( \text{\textsc{W}}\right) $}\smallskip ,

where $\mathcal{R}^{\prime }\left( \partial _{1}^{\prime },\partial
_{2}\right) $ is defined by induction on $h\left( \partial _{2}\right) $ --
either trivially, if $\partial _{2}=\left( \text{\textsc{Ax}}\right) ^{+}$,
or hereditarily, in the non-principal subcases, while in the principal
subcases\smallskip

$\partial _{2}:\ $\fbox{$\dfrac{\left( \partial _{2}^{\prime }:\overline{A},%
\overrightarrow{G}\right) }{%
\begin{array}{c}
\left[ p\right] A,\left\langle p\right\rangle \!\overrightarrow{G},\Pi
^{\prime } \\ 
=\left[ p\right] \!A,\Pi
\end{array}
}\left( \text{\textsc{Gen}}\right) $} \ and $\partial _{2}:\ $\fbox{$\dfrac{%
\left( \partial _{2}^{\prime }:\overrightarrow{G},H\right) }{%
\begin{array}{c}
\left[ p\right] A,\left\langle p\right\rangle \!\overrightarrow{G},\left[ p%
\right] \!H,\Pi ^{\prime \prime } \\ 
=\left[ p\right] \!A,\Pi
\end{array}
}\left( \text{\textsc{Gen}}\right) $}\smallskip

we respectively let\smallskip

$\mathcal{R}^{\prime }\left( \partial _{1},\partial _{2}\right) \smallskip $%
\smallskip $:=\ $\fbox{$\dfrac{\left( \partial _{1}^{\prime }:A,%
\overrightarrow{B},D\right) \quad \quad \left( \partial _{2}^{\prime }:%
\overline{A},\overrightarrow{G}\right) }{\dfrac{\overrightarrow{B},D,%
\overrightarrow{G}}{%
\begin{array}{c}
\left\langle p\right\rangle \!\overrightarrow{B},\left[ p\right]
\!D,\left\langle p\right\rangle \!\overrightarrow{G},\Pi ^{\prime } \\ 
=\left\langle p\right\rangle \!\overrightarrow{B},\left[ p\right] \!D,\Pi
\end{array}
}\left( \text{\textsc{Gen}}\right) }\left( \text{\textsc{Cut}}\right) $}%
\smallskip\ \ and

$\mathcal{R}^{\prime }\left( \partial _{1},\partial _{2}\right) \smallskip
\smallskip :=\ \fbox{$\dfrac{\left( \partial _{2}^{\prime }:\overrightarrow{G%
},H\right) }{%
\begin{array}{c}
\left\langle p\right\rangle \!\overrightarrow{G},\left[ p\right]
\!H,\left\langle p\right\rangle \!\overrightarrow{B},\left[ p\right] \!D,\Pi
^{\prime \prime } \\ 
=\left\langle p\right\rangle \!\overrightarrow{B},\left[ p\right] \!D,\Pi
\end{array}
}\left( \text{\textsc{Gen}}\right) $}$, as desired.$\smallskip $

Obviously $\mathcal{R}$ reduces the cut degree of $\partial $. That is, in
each case 1--5 we have$\ \deg \left( \mathcal{R}\left( \partial \right)
\right) <\deg \left( \partial \right) <\omega ^{\omega }$, provided that
both $\partial _{1}$ and $\partial _{2}$ involved are cutffree. Moreover
it's readily seen that nodes in $\mathcal{R}\left( \partial \right) $ can be
augmented with ordinals such that 
\begin{equation*}
\fbox{$h\left( \mathcal{R}\left( \partial \right) \right) <h\left( \partial
_{1}\right) +\!\!\!\!\!+\,\,h\left( \partial _{2}\right) +\omega <h\left(
\partial \right) \cdot 2+\omega $}\text{.}
\end{equation*}
Having this one can define ordinal assignments also for (slightly modified)
cutfree derivations $\mathcal{E}\left( \partial \right) $ such that for any $%
\partial $\ with $\deg \left( \partial \right) <\omega ^{\alpha }$ it holds 
\begin{equation*}
\fbox{$h\left( \mathcal{E}\left( \partial \right) \right) <\varphi \left(
\alpha ,h\left( \partial \right) \right) $}\text{,}
\end{equation*}
which for $\deg \left( \partial \right) <\omega ^{\omega }$ and $h\left(
\partial \right) <\omega \cdot 2$ (cf. Theorem 4) yields 
\begin{equation*}
\fbox{$h\left( \mathcal{E}\left( \partial \right) \right) <\underset{%
n<\omega }{\sup }\,\varphi \left( n,\omega \cdot 2\right) =\varphi \left(
\omega ,0\right) =\varphi _{\omega }\!\left( 0\right) $}
\end{equation*}
(see Appendix A for a detailed presentation). It is readily seen that the
entire proof is formalizable in $\mathbf{PA}_{\varphi _{\omega }\left(
0\right) }$, i.e. $\mathbf{PA}$ extended by schema of transfinite induction
along (canonical primitive recursive representation of) ordinal $\varphi
_{\omega }\!\left( 0\right) $. \footnote{$\varphi _{\omega }\!\left(
0\right) =D\!\left( \omega ^{\Omega +\omega }\right) ${\footnotesize \
according to} {\footnotesize ordinal notations used in \cite{Gordeev}.}}
\end{proof}

\begin{corollary}
Let $\Gamma $ be any sequent that does not contain occurrences $\left[
P^{\ast }\right] $ and suppose that $\Gamma $ is derivable in \textsc{Seq}$%
_{\omega }^{\text{\textsc{pdl}}}$. Then $\Gamma $ is derivable in a
subsystem of \textsc{Seq}$_{\omega }^{\text{\textsc{pdl}}}$, called \textsc{%
Seq}$_{1}^{\text{\textsc{pdl}}}$, that does not contain inferences $\left[
\ast \right] $ and/or $\left( \text{\textsc{Cut}}\right) $. Note that every
derivation in \textsc{Seq}$_{1}^{\text{\textsc{pdl}}}$ is finite.
Consequently, any given $\left[ P^{\ast }\right] $-free seq-formula is valid
in \textbf{PDL} iff it is derivable in \textsc{Seq}$_{1}^{\text{\textsc{pdl}}%
}$.
\end{corollary}

\begin{proof}
This is obvious by the subformula property of cutfree derivations.
\end{proof}

\begin{remark}
Here and below we argue in $\mathbf{PA}_{\varphi _{\omega }\left( 0\right) }$
that is a proper extension of $\mathbf{PA}$, as $\varphi _{\omega }\left(
0\right) >\varepsilon _{0}$. Actually by standard arguments the whole proof
is formalizable in the corresponding primitive recursive weakening, $\mathbf{%
PRA}_{\varphi _{\omega }\left( 0\right) }$.
\end{remark}

\subsection{Herbrand-style conclusions}

Let $\mathcal{L}_{0}$ be the star-free sublanguage of $\mathcal{L}$. Denote
by \textsc{Seq}$_{0}^{\text{\textsc{pdl}}}$ the $\mathcal{L}_{0}$-subsystem
of \textsc{Seq}$_{1}^{\text{\textsc{pdl}}}$.

\begin{theorem}
Let $\Sigma =\left\langle P^{\ast }\right\rangle \!A,\Pi $ with $A,\Pi \in 
\mathcal{L}_{0}$. Suppose that $\Sigma $ is derivable in \textsc{Seq}$%
_{\omega }^{\text{\textsc{pdl}}}$. Then there exists a $k\geq 0$ such that $%
\widehat{\Sigma }_{k}:=A,\left\langle p\right\rangle \!A,\cdots
,\left\langle p\right\rangle ^{k}\!\!A,\Pi $ is derivable in \textsc{Seq}$%
_{0}^{\text{\textsc{pdl}}}$.
\end{theorem}

\begin{proof}
The nontrivial implication $\,$\textsc{Seq}$_{\omega }^{\text{\textsc{pdl}}}$
$\vdash \Sigma \Rightarrow \,\,\text{\textsc{Seq}}_{0}^{\text{\textsc{pdl}}%
}\vdash \widehat{\Sigma }_{k}$ follows by standard arguments from the cut
elimination theorem by induction on the height of the corresponding finite
cutfree proof $\partial $ of $\Sigma $\ in \textsc{Seq}$_{\omega }^{\text{%
\textsc{pdl}}}$. Since no $\left[ P^{\ast }\right] $ occurs in $\Sigma $, no 
$\left\langle P^{\ast }\right\rangle \!A$ can be principal formula in $%
\left( \text{\textsc{Gen}}\right) $. Thus the only crucial case is when some 
$\left\langle P^{\ast }\right\rangle \!A$ is principal formula in 
\begin{equation*}
\fbox{$\left\langle \ast \right\rangle \quad \dfrac{\!\!\!\left\langle
P^{\ast }\right\rangle \!A,\left\langle P\right\rangle ^{m}\!\!A,\Sigma }{%
\left\langle P^{\ast }\right\rangle \!A,\Sigma }\ $}
\end{equation*}
which by the induction hypothesis yields $k$ such that$\ \left\langle
P\right\rangle ^{m}\!\!A,\widehat{\Sigma }_{k}$ is derivable in \textsc{Seq}$%
_{0}^{\text{\textsc{pdl}}}$. By $\left( \text{\textsc{C}}\right) $ or $%
\left( \text{\textsc{W}}\right) $ this yields the derivability of $\widehat{%
\Sigma }_{k^{\prime }}$\ for $k^{\prime }:=\max \left( k,m\right) $.
\end{proof}

\begin{remark}
By the same token, for any $\left[ P^{\ast }\right] $-free seq-formula $F$,
one can successively replace all subformulas $\left\langle P^{\ast
}\right\rangle \!A$ by appropriate disjunctions $\overset{k}{\underset{i=0}{%
\bigvee }}\left\langle P\right\rangle ^{i}\!\!A$ such that $F$ is $\mathbf{%
PDL}$-valid iff the resulting expansion $\widehat{F}$ is derivable in 
\textsc{Seq}$_{0}^{\text{\textsc{pdl}}}$.
\end{remark}

\subsubsection{PSPACE refinement}

Denote by $\mathcal{L}_{00}$ a sublanguage of $\mathcal{L}_{0}$ over atomic
programs \textrm{PRO}$_{0}$ and let $\mathcal{L}_{\emptyset }$ be purely
propositional fraction of $\mathcal{L}.$ Note that program operations ``$;$%
'' and ``$\cup $'' are definable in $\mathcal{L}_{00}$ via $\left(
P;Q\right) \!A:=\left( P\right) \!\left( Q\right) \!A$, $\left\langle P\cup
Q\right\rangle \!A:=\left\langle P\right\rangle \!A\vee \left\langle
Q\right\rangle \!A$\ and $\left[ P\cup Q\right] \!A:=\left[ P\right]
\!A\wedge \left[ Q\right] \!A$. Let \textsc{Seq}$_{00}^{\text{\textsc{pdl}}}$
be the following $\mathcal{L}_{00}$-restriction of \textsc{Seq}$_{0}^{\text{%
\textsc{pdl}}}$ (that proves the same $\mathcal{L}_{00}$-sequents as \textsc{%
Seq}$_{0}^{\text{\textsc{pdl}}}$).

\begin{equation*}
\begin{array}{c}
\quad 
\begin{array}{c}
\fbox{$\left( \text{\textsc{Ax}}\right) \quad x,\lnot x,\Gamma $}
\end{array}
\quad \quad \qquad \\ 
\begin{array}{c}
\fbox{$\left( \vee \right) \quad \dfrac{A,B,\Gamma }{A\vee B,\Gamma }\ $}%
\fbox{$\left( \wedge \right) \quad \dfrac{A,\Gamma \quad \quad B,\Gamma }{%
A\wedge B,\Gamma }\ $}
\end{array}
\\ 
\qquad \qquad \text{\ }\fbox{$
\begin{array}{c}
\left( \text{\textsc{Gen}}\right) \quad \dfrac{A_{1},\cdots ,A_{n}}{\left(
p\right) _{\chi _{1}}\!\!A_{1},\cdots ,\left( p\right) \!_{\chi
_{n}}\!A_{n},\Gamma }\ \left( n>0\right) \\ 
\text{if }\overset{n}{\underset{i=1}{\sum }}\chi _{i}=1.
\end{array}
$}\qquad \qquad \qquad
\end{array}
\end{equation*}

\begin{remark}
Theorems 4 and 8 confirm that \textbf{PDL} is a conservative extension of
classical propositional logic that is formalized by the $\left( \text{%
\textsc{Gen}}\right) $-free subsystem of \textsc{Seq}$_{00}^{\text{\textsc{%
pdl}}}$. Thus any $\mathcal{L}_{\emptyset }$-formula $A$ is derivable in 
\textsc{Seq}$_{00}^{\text{\textsc{pdl}}}-\left( \text{\textsc{Gen}}\right) $
iff it is valid in propositional logic, and hence, by contraposition, 
\textsc{Seq}$_{00}^{\text{\textsc{pdl}}}\nvdash A$ iff $\nvDash A$ (: $\lnot
A$ is satisfiable).
\end{remark}

\begin{lemma}[$p\,$-inversion]
\ Suppose that $\left[ p\right] \!A_{1},\cdots ,\left[ p\right]
\!A_{j},\left\langle p\right\rangle \!B_{1},\cdots ,\left\langle
p\right\rangle \!B_{k},\Gamma $, where $\Gamma \!=\!\left( q_{1}\right)
\!C_{1},\cdots ,\left( q_{l}\right) \!C_{l},\Pi $ for $q_{j}\neq p$, and $%
\Pi \in \mathcal{L}_{00}$, is derivable in \textsc{Seq}$_{00}^{\text{\textsc{%
pdl}}}$. Then so is either $\Gamma $ or $A_{i},B_{1},\cdots ,B_{k}$, for
some $i\in \left[ 1,j\right] $, without increasing the height of the former
derivation.
\end{lemma}

\begin{proof}
By straightforward induction on the derivation height. In the crucial
principal case we have 
\begin{equation*}
\fbox{$\ \left( \text{\textsc{Gen}}\right) \quad \dfrac{A_{i},\Delta }{\left[
p\right] \!A_{1},\cdots ,\left[ p\right] \!A_{k},\left\langle p\right\rangle
\!B_{1},\cdots ,\left\langle p\right\rangle \!B_{l}}$}
\end{equation*}
\smallskip

where $0<i\leq k$ and $\Delta \subseteq B_{1},\cdots ,B_{l}$, which by
derivable $\left( \text{\textsc{W}}\right) $ yields\smallskip\ the required
derivability of $A_{i},B_{1},\cdots ,B_{l}$.
\end{proof}

\begin{theorem}
The derivability in \textsc{Seq}$_{00}^{\text{\textsc{pdl}}}$\ is a PSPACE
problem.\ \footnote{{\footnotesize Apparently this result is well-known in
the context of multimodal version of \textbf{K}.}}
\end{theorem}

\begin{proof}
For the sake of brevity we consider $\mathcal{L}_{00}$ formulas containing
at most one atomic program $p=\pi _{0}$. Furthermore, we refine the notion
of \textsc{Seq}$_{00}^{\text{\textsc{pdl}}}$\ derivability by asserting that
a sequent $\Delta \neq \left( \text{\textsc{Ax}}\right) $ is the conclusion
of a rule $\left( \text{\textsc{R}}\right) $ if one of the following \emph{%
priority conditions} 1--3 is satisfied.

\begin{enumerate}
\item  $\left( \text{\textsc{R}}\right) =\left( \vee \right) $.

\item  $\left( \text{\textsc{R}}\right) =\left( \wedge \right) $ and no
disjunction $A\vee B$ occurs as formula in $\Delta $; thus $\Delta $ is not
a conclusion of any $\left( \vee \right) $.

\item  No disjunction $A\vee B$ or conjunction $A\wedge B$ occurs as formula
in $\Delta $. Thus $\Delta $ is not a conclusion of any $\left( \vee \right) 
$ or $\left( \wedge \right) $, i.e. $\Delta =\left( p\right) _{\xi
_{1}}\!\!F_{1},\cdots ,\left( p\right) \!_{\xi _{n}}\!F_{n}$ for $\overset{n%
}{\underset{i=1}{\sum }}\!\,\xi _{i}\geq 1$. In this case we stipulate that $%
\Delta $ is the conclusion of $\left( \text{\textsc{R}}\right) $ if one of
the following two conditions holds:

\begin{enumerate}
\item  $\overset{n}{\underset{i=1}{\sum }}\!\,\xi _{i}=1$ and $F_{1},\cdots
,F_{n}$ is the premise of $\left( \text{\textsc{R}}\right) =\left( \text{%
\textsc{Gen}}\right) $.

\item  $\overset{n}{\underset{i=1}{\sum }}\!\,\xi _{i}>1$ and there exists $%
j\in \left[ 1,n\right] $ with $\xi _{j}=1$ such that either $\Delta ^{\left(
j\right) }:=F_{j}\cup \left\{ F_{l}\in \Delta :\xi _{l}=0\right\} $ or $%
\Delta ^{\left( -j\right) }:=\Delta \setminus \left\{ F_{j}\right\} $ is the
premise of $\left( \text{\textsc{R}}\right) $. (Note that we have $\left( 
\text{\textsc{R}}\right) =\left( \text{\textsc{Gen}}\right) $ and $\left( 
\text{\textsc{R}}\right) =\left( \text{\textsc{W}}\right) $ in the former
and in the latter case, respectively.)
\end{enumerate}
\end{enumerate}

Having this we consider derivations in the refined \textsc{Seq}$_{00}^{\text{%
\textsc{pdl}}}$ as at most binary-branching trees $\partial $ whose nodes
are labeled with sequents of $\mathcal{L}_{00}$. Actually, for any given $%
\mathcal{L}_{00}$-sequent $\Sigma $ it will suffice to fix one distinguished 
\emph{proof search tree} $\partial _{0}$\ with root sequent $\Sigma $ that
is defined by bottom-up recursion while applying the conditions 1--3 in a
chosen order as long as possible. It is readily seen by inversions in
Lemmata 7, 14 that $\Sigma $ is derivable in \textsc{Seq}$_{00}^{\text{%
\textsc{pdl}}}$ iff $\partial _{0}$ proves $\Sigma $, i.e. every maximal
path in $\partial _{0}$ is locally correct with respect to 1--3. Moreover,
by the obvious subformula property we conclude that the depth, $d\left(
\partial _{0}\right) $, and maximum sequent length, $\max \left\{ \left|
\Delta \right| :\Delta \in \partial _{0}\right\} $, of $\partial _{0}$ are
both proportional to $\left| \Sigma \right| $. Hence every maximal path in $%
\partial _{0}$ can be encoded by a $\mathcal{L}_{01}$-string of the length
proportional to $\left| \Sigma \right| $ whose local correctness is
verifiable by TM in $\mathcal{O}\left( \left| \Sigma \right| \right) $\
space. The corresponding universal verification runs by counting all maximal
paths successively, still in $\mathcal{O}\left( \left| \Sigma \right|
\right) $\ space, which completes the proof.
\end{proof}

\begin{remark}
Arguing along more familiar lines we can turn $\partial _{0}$ into a Boolean
circuit with (binary) AND, OR and (unary) ID gates, where ID$\left( x\right)
:=x$ for $x\in \left\{ 0,1\right\} $, such that AND, OR and ID correspond to
the above conditions 2, 3 (b) and 1 and/or 3 (a),\ respectively. The
corresponding truth evaluations $val\left( -\right) $ are defined as usual
via $val\left( \Delta \right) :=1$ $\left( \mathbf{true}\right) $ iff $%
\Delta =\left( \text{\textsc{Ax}}\right) $, for every leaf $\Delta $. Then $%
val\left( \Sigma \right) =1$ iff $\partial _{0}$ proves $\Sigma $, as
required. \footnote{{\footnotesize This proof is dual to familiar proof of
polynomial space solvability of QSAT (cf. e,g. \cite{Papa}).}}
\end{remark}

\subsubsection{Special cases}

Recall that by (a particular case of) Theorem 11, for any $\Sigma
=\left\langle p^{\ast }\right\rangle \!A,\Pi $ with $A\in \mathcal{L}%
_{00},\Pi \in \mathcal{L}_{\emptyset }$ the following holds. Suppose that $%
\Sigma $ is derivable in \textsc{Seq}$_{\omega }^{\text{\textsc{pdl}}}$.
Then there exists a $k\geq 0$ such that $\widehat{\Sigma }%
_{k}:=A,\left\langle p\right\rangle \!A,\cdots ,\left\langle p\right\rangle
^{k}\!\!A,\Pi $ is derivable in \textsc{Seq}$_{00}^{\text{\textsc{pdl}}}$.
It turns out that in some cases it's possible to estimate the minimum $k$
and hence the size of $\widehat{\Sigma }_{k}$.

\begin{definition}
Let $p=\pi _{0}$ be fixed. Call \emph{basic conjunctive normal form} (abbr.: 
\emph{BCNF}) any $\mathcal{L}_{00}$-formula$\!$ $\overset{m}{\underset{i=1}{%
\bigwedge }}\left( B_{i}\vee \left\langle p\right\rangle \!C_{i}\vee 
\overset{n_{i}}{\underset{j=1}{\bigvee }}\left[ p\right] \!D_{i,j}\right) $
for $m>0$, $n_{i}\geq 0$ and $B_{i}$, $C_{i}$, $D_{i,j}$ $\in \mathcal{L}%
_{\emptyset }\cup \left\{ \emptyset \right\} $. Formulas $\left\langle
p^{\ast }\right\rangle \!A\vee Z$ for $A\in \mathrm{BCNF}$ and $Z\in \!%
\mathcal{L}_{\emptyset }$ are called \emph{basic conjunctive normal
expressions} (abbr.: \emph{BCNE}).
\end{definition}

\begin{theorem}
Let $A=\overset{m}{\underset{i=1}{\bigwedge }}\left( \!B_{i}\vee
\left\langle p\right\rangle \!C_{i}\vee \overset{n_{i}}{\underset{j=1}{%
\bigvee }}\left[ p\right] \!D_{i,j}\!\right) \!\in $ $\mathrm{BCNF}\!$ and
for any $k\geq 0$ and $\Pi \!\in \!\mathcal{L}_{\emptyset }$ let $\widehat{A}%
_{k}:=$ $A,\left\langle p\right\rangle \!A,\cdots ,\left\langle
p\right\rangle ^{k}\!\!A$, $\widehat{A}_{0}=A$, and $\widehat{\Sigma }_{k}:=%
\widehat{A}_{k},\Pi $. \negthinspace If any $\widehat{\Sigma }_{k}$
\negthinspace is \negthinspace derivable in \textsc{Seq}$_{00}^{\text{%
\textsc{pdl}}}$ then so is $\widehat{\Sigma }_{n+1}$ too, where $n=\overset{m%
}{\underset{i=1}{\sum }}n_{i}$.
\end{theorem}

\begin{proof}
For $i\in \left[ 1,m\right] $ let $\Delta _{i}:=$ $\left\{ \left[ p\right]
D_{i,j}\!:1\leq j\leq n_{i}\right\} $. So Lemma 14 yields 
\begin{equation*}
\vdash \,\!\!\widehat{\Sigma }_{0}\Leftrightarrow \,\vdash A,\Pi
\Leftrightarrow \overset{m}{\underset{i=1}{\bigwedge }}\!\vdash
B_{i},\left\langle p\right\rangle \!C_{i},\Delta _{i},\Pi \Leftrightarrow
\end{equation*}
\begin{equation*}
\fbox{$\left. 
\begin{array}{c}
``F\!or\ every\ i\in \!\left[ 1,m\right] \!,\,either\vdash B_{i},\Pi \ or \\ 
there\ is\ j\in \!\left[ 1,n_{i}\right] \ such\ that\vdash C_{i},D_{i,j}"\,
\end{array}
\right. $}
\end{equation*}

where ``$\ \vdash $'' stands for ``\thinspace \thinspace \textsc{Seq}$_{00}^{%
\text{\textsc{pdl}}}\vdash $'', and hence 
\begin{equation*}
\nvdash \ \!\!\widehat{\Sigma }_{0}\Leftrightarrow \,\nvdash A,\Pi
\Leftrightarrow
\end{equation*}
\begin{equation*}
\fbox{$\left. 
\begin{array}{c}
``T\!here\ is\ i\in \!\left[ 1,m\right] \!\ such\ that\nvdash B_{i},\Pi \\ 
and\ f\!or\ every\ j\in \!\left[ 1,n_{i}\right] ,\,\nvdash C_{i},D_{i,j}"\,.
\end{array}
\right. $}
\end{equation*}
By the same token, for any $s\geq 0$ we let $\left\langle p\right\rangle 
\widehat{\!A}_{s}:=\left\langle p\right\rangle \left( A\vee \left\langle
p\right\rangle \!A\vee \cdots \vee \left\langle p\right\rangle
^{s}\!\!A\right) $ and arrive at 
\begin{eqnarray*}
&\vdash &\!\!\!\!\widehat{\Sigma }_{s+1}\Leftrightarrow \ \vdash \ \widehat{A%
}_{s+1},\Pi \Leftrightarrow \ \vdash A,\left\langle p\right\rangle \widehat{%
\!A}_{s},\Pi \\
&\Leftrightarrow &\overset{m}{\underset{i=1}{\bigwedge }}\!\vdash
B_{i},\left\langle p\right\rangle \!C_{i},\Delta _{i},\left\langle
p\right\rangle \!\widehat{A}_{s},\Pi \Leftrightarrow
\end{eqnarray*}
\begin{equation*}
\fbox{$\left. 
\begin{array}{c}
``F\!or\ every\ i\in \!\left[ 1,m\right] \!,\,either\vdash B_{i},\Pi \ or\ 
\\ 
there\ is\ j\in \!\left[ 1,n_{i}\right] \ such\ that\vdash \widehat{A}%
_{s},C_{i},D_{i,j}"\,
\end{array}
\right. $}
\end{equation*}
which yields 
\begin{equation*}
\nvdash \ \!\widehat{\Sigma }_{s+1}\Leftrightarrow \,\nvdash \widehat{A}%
_{s+1},\Pi \Leftrightarrow
\end{equation*}
\begin{equation*}
\fbox{$\left. 
\begin{array}{c}
``T\!here\ is\ i\in \!\left[ 1,m\right] \!\ such\ that\nvdash B_{i},\Pi \ \ 
\\ 
and\ f\!or\ every\ j\in \!\left[ 1,n_{i}\right] ,\,\nvdash \widehat{A}%
_{s},C_{i},D_{i,j}"\,.
\end{array}
\right. $}
\end{equation*}
Thus for any $k\geq 0$, the assertion $\nvdash \widehat{\Sigma }_{k}$ is
equivalent to the existence of a labeled rooted \emph{refutation} \emph{tree}
$T_{k}$ of the height $k+1$ such that the following conditions 1--3 hold,
where sequents $\ell \left( x\right) $ are the labels of nodes $x\in T_{k}$ (%
$\rho $ being the root).

\begin{enumerate}
\item  $\ell \left( \rho \right) =\Pi $.

\item  $\nvdash \ell \left( x\right) $ holds for every leaf $x\in T_{k}$.

\item  For any inner node $x\in T_{k}$ there exists $i\in \left[ 1,m\right] $
such that $x$ has $m_{i}+1$ ordered children: $x_{0}$ (the \emph{son}) with
label $\ell \left( x_{0}\right) =B_{i},\ell \left( x\right) $ and $%
x_{1},\cdots ,x_{m_{i}}\!$ (the \emph{daughters}) labeled $\ell \left(
x_{j}\right) =C_{i},D_{i,j}$, respectively; moreover $x_{j}$ ($j\geq 0$) is
a leaf iff it is either a son or else a daughter of the depth $k+1$.
\end{enumerate}

Such $T_{k}$ is easily obtained by straightforward geometric interpretation
of our translations of the conditions $\nvdash \ \!\!\widehat{\Sigma }_{0}$
and $\nvdash \ \!\widehat{\Sigma }_{s+1}$. Moreover, condition 2 above is
equivalent to

\begin{description}
\item  \ \thinspace \thinspace 2*. $\nvdash \ell \left( x\right) $ holds for
every node $x\in T_{k}$,
\end{description}

since daughters are subsequents of their sons. That is, every inner node $%
x\in T_{k}$ with label $\ell \left( x\right) $ has an upper neighbor (son) $%
x_{0}$ with label $B_{i},\ell \left( x\right) $ being a leaf in $T_{k}$,
which by condition 2 yields $\nvdash B_{i},\ell \left( x\right) $, and hence
also $\nvdash \ell \left( x\right) $, as required, by converting the
admissible weakness rule of \textsc{Seq}$_{01}^{\text{\textsc{pdl}}}$.

Now if $k\leq n+1$ then $\widehat{\Sigma }_{k}\subseteq \widehat{\Sigma }%
_{n+1}$, and hence $\vdash \widehat{\Sigma }_{k}$ implies $\vdash \widehat{%
\Sigma }_{n+1}$. Furthermore, assuming $\nvdash \widehat{\Sigma }_{n+1}$
we'll infer $\left( \forall s>n\right) \nvdash \widehat{\Sigma }_{s}$ and
conclude by contraposition that $\left( \exists k\right) \vdash \widehat{%
\Sigma }_{k}$ implies (in fact is equivalent to) $\vdash \widehat{\Sigma }%
_{n+1}$, as required. So assume $\nvdash \widehat{\Sigma }_{n+1}$. We prove
by induction on $s>n$ the existence of the refutation trees $T_{s}$, and
hence $\nvdash \widehat{\Sigma }_{s}$. Basis case $s=n+1$ holds by the
assumption. To pass from $T_{s}$ to $T_{s+1}$ we argue as follows. Let $x\in
T_{s}$ be any leaf-daughter and $\theta =\left( \rho ,y_{1},\cdots
,y_{s}=x\right) $ the corresponding maximal path, in $T_{s}$. Since $\theta $
contains at most $n<s$ different labels $\ell \left( y_{\imath }\right)
=C_{i},D_{i,j}$ ($i\in \left[ 1,m\right] $, $j\in \left[ 1,n_{i}\right] $),
there exist a (say, minimal) pair $0<r<t<s$ such that $\ell \left(
y_{r}\right) =\ell \left( y_{t}\right) $. Let $T_{\left( s,x,r,t\right) }$
be a tree that arises from $T_{s}$ by substituting its subtree rooted in $%
y_{r}$ for that rooted in $y_{t}$. Note that $T_{\left( s,x,r,t\right) }$ is
higher than $T_{s}$ -- so let $T_{s+1}^{\left( x\right) }$ be a subtree of $%
T_{\left( s,x,r,t\right) }$\ consisting of the nodes of heights $\leq s+1$.
Proceeding this way successively with respect to all leaf-daughters $x\in
T_{s}$ while keeping in mind condition 2*, we eventually obtain a refutation
tree $T_{s+1}$ of the height $s+1$, as required.
\end{proof}

By Remark 10 and Theorem 11, the following are provable in $\mathbf{PRA}%
_{\varphi _{\omega }\left( 0\right) }$.

\begin{corollary}
Let $A\in $ $\mathrm{BCNF}\!$, $n$ and $\Pi $ be as above. Then
\negthinspace $\Sigma :=\left\langle p^{\ast }\right\rangle \!A,\Pi $
\negthinspace is \negthinspace derivable in \textsc{Seq}$_{\omega }^{\text{%
\textsc{pdl}}}$ \negthinspace iff $\widehat{\Sigma }_{n+1}:=\widehat{A}%
_{n+1},\Pi $ is derivable in \textsc{Seq}$_{00}^{\text{\textsc{pdl}}}$.
\end{corollary}

\begin{corollary}
Let $S\in $ $\mathrm{BCNE}$. \negthinspace Problem $\mathbf{PDL}\vdash S$,
i.e. $\mathbf{PDL}$-validity of $S$, is solvable by a deterministic TM in $%
\mathcal{O}\left( \left| S\right| ^{2}\right) $ space.
\end{corollary}

\begin{proof}
For $A$ as above we have $n<\left| A\right| $, and hence$\left| \widehat{A}%
_{n+1}\right| =\mathcal{O}\left( \left| A\right| ^{2}\right) $. This yields $%
\left| \widehat{A}_{n+1},Z\right| =\mathcal{O}\left( \left| A\right|
^{2}+\left| Z\right| \right) =\mathcal{O}\left( \left| S\right| ^{2}\right) $%
. Now by Theorem 4 followed by Theorems 11, 18 we have 
\begin{equation*}
\fbox{$\mathbf{PDL}\vdash S\Leftrightarrow \text{\textsc{Seq}}_{\omega }^{%
\text{\textsc{pdl}}}\vdash S$ $\Leftrightarrow $ \textsc{Seq}$_{01}^{\text{%
\textsc{pdl}}}\vdash \widehat{A}_{n+1},Z$}
\end{equation*}
while problem \textsc{Seq}$_{01}^{\text{\textsc{pdl}}}\vdash \widehat{A}%
_{n+1},Z$ is solvable in $\mathcal{O\!}\left( \left| \widehat{A}%
_{n+1},Z\right| \right) \!=\!\mathcal{O\!}\left( \left| S\right| ^{2}\right) 
$ space.
\end{proof}

Now consider (dual) basic disjunctive normal forms.

\begin{definition}
Call \emph{basic disjunctive normal form} (abbr.: \emph{BDNF}) any $\mathcal{%
L}_{00}$-formula$\!$ $F\vee \overset{s}{\underset{i=1}{\bigvee }}\left(
F_{i}\wedge \left[ p\right] G_{i}\right) \vee \overset{t}{\underset{j=1}{%
\bigvee }}\left( F_{j}\wedge \left\langle p\right\rangle \!H_{j}\right) $
for $s,t>0$ and $F$, $F_{i}$, $G_{i}$, $H_{j}$ $\in \mathcal{L}_{\emptyset
}\cup \left\{ \emptyset \right\} $. Formulas $\left\langle p^{\ast
}\right\rangle \!A\vee Z$ for $A\in \mathrm{BDNF}$ and $Z\in \!\mathcal{L}%
_{\emptyset }$ are called \emph{basic disjunctive normal expressions}
(abbr.: \emph{BDNE}).
\end{definition}

\begin{problem}
Let $S\in \mathrm{BDNE}$. \negthinspace Is problem $\mathbf{PDL}\vdash S$
solvable by a TM in $\left| S\right| $-polynomial space?
\end{problem}

\subsubsection{More on BDNE}

\textbf{PDL}-satisfiability problem for certain statements \textsc{Accepts}$%
_{M,x}=\left[ p^{\ast }\right] V\wedge W$ for $V\in $ $\mathrm{BCNF}\!$, $%
W\in \mathcal{L}_{\emptyset }$ -- expressing that satisfying Kripke frames
encode accepting computations of polynomial-space alternating TM -- is known
to be EXPTIME-complete (cf. \cite{ModalLogic} and \cite{DynamicLogic}:
Theorem 8.5, et al; see also \cite{Spaan}). Hence so is also the \textbf{PDL}%
-validity problem for the corresponding negations $S:=\overline{\text{%
\textsc{Accepts}}}_{M,x}=\left\langle p^{\ast }\right\rangle \!A\vee Z\in $ $%
\mathrm{BDNE}$. So the affirmative solution to Problem 22 would infer $%
\mathbf{EXPTIME=PSPACE}$ (and vice versa, since general \textbf{PDL}%
-validity is EXPTIME-complete). That is, problem $\mathbf{EXPTIME=PSPACE}$
reduces to a particular case of Problem 22 for $S:=\overline{\text{\textsc{%
Accepts}}}_{M,x}$ (see Appendix B for precise definition).

Now let $S=\left\langle p^{\ast }\right\rangle \!A\vee Z\in $ $\mathrm{BDNE}$
for $A\!=\!F\vee \overset{s}{\underset{i=1}{\bigvee }}\!\left( F_{i}\!\wedge
\!\left[ p\right] G_{i}\right) \vee \overset{t}{\underset{j=1}{\bigvee }}%
\!\left( F_{j}\!\wedge \!\left\langle p\right\rangle \!H_{j}\right) \!\in 
\mathrm{BDNF}$ and $Z\in \!\mathcal{L}_{\emptyset }$. We wish to present the
assertion $\mathbf{PDL}\vdash S$ in a suitable ``transparent'' quantified
boolean form. To this end, by DeMorgan laws, we first convert $A\!$\ to $R=%
\underset{\xi \in \Xi }{\bigwedge }R_{\xi }\in \mathrm{BCNF}$, where $R_{\xi
}=B_{\xi }\vee \left\langle p\right\rangle \!C_{\xi }\vee \underset{\jmath
\in J_{\xi }}{\bigvee }\left[ p\right] \!D_{\xi ,\jmath }$ for $\Xi
:=\left\{ \xi =\left( \xi \left( 1\right) ,\cdots ,\xi \left( s+t\right)
\right) \right\} $ with $\xi \left( k\right) \in \left\{ 1,2\right\} $, $%
1\leq k\leq s+t$, while 
\begin{eqnarray*}
B_{\xi } &:&=F\vee \bigvee \left\{ F_{k}:1\leq k\leq s+t\wedge \xi \left(
k\right) =1\right\} , \\
C_{\xi } &:&=\bigvee \left\{ H_{k-s}:s<k\leq t\wedge \xi \left( k\right)
=2\right\} , \\
D_{\xi ,\jmath } &:&=G_{\jmath }\text{ for }\jmath \in J_{\xi }:=\left\{
k:1\leq k\leq s\wedge \xi \left( i\right) =2\right\} .\text{\ }
\end{eqnarray*}
Clearly $\mathbf{PDL}\vdash A\leftrightarrow R$ (also by \textbf{PDL}%
-equivalence $\left\langle p\right\rangle \!H\vee \left\langle
p\right\rangle \!H^{\prime }\leftrightarrow $ $\left\langle p\right\rangle
\!\left( H\vee \!H^{\prime }\right) \,$). Note that $\left| \Xi \right|
=2^{s+t}$ and $\left| R_{\xi }\right| <\left| A\right| $, for every $\xi \in
\Xi $.

By the cut-elimination theorem, $\mathbf{PDL}\vdash S$\ is equivalent to 
\textsc{Seq}$_{\omega }^{\text{\textsc{pdl}}}\vdash \left\langle p^{\ast
}\right\rangle \!R,Z$, which by Theorem 18 is equivalent to \textsc{Seq}$%
_{00}^{\text{\textsc{pdl}}}\vdash \widehat{R}_{n+1},Z$, where 
\begin{equation*}
\widehat{R}_{n+1}=R,\left\langle p\right\rangle \!R,\cdots ,\left\langle
p\right\rangle ^{n+1}\!R
\end{equation*}
for $n:=\underset{\xi \in \Xi }{\sum }\left| J_{\xi }\right| <s\cdot \left|
\Xi \right| =s2^{s+t}$. Arguing as in the proof of Theorem 18 we get 
\begin{equation*}
\fbox{$\mathbf{PDL}\vdash S$\ $\Leftrightarrow \ $\textsc{Seq}$_{00}^{\text{%
\textsc{pdl}}}\vdash \widehat{R}_{n+1},Z\Leftrightarrow f\!\left(
s2^{t}+1,Z\right) =1$}
\end{equation*}
where $f$ is a boolean-valued binary function that is defined for every $%
\imath \geq 0$ and propositional formula $X$ by the following recursive
clauses 1--2, where ``$\,\vdash _{\emptyset }Y$ '' stands for plain boolean
validity of propositional formula$\ Y$.

\begin{enumerate}
\item  \noindent \fbox{$f\!\left( 0,X\right) =1\Leftrightarrow \underset{\xi
\in \Xi }{\bigwedge }\left( \vdash _{\emptyset }\left( B_{\xi }\vee X\right)
\ or\ \underset{\text{ }j\in _{\xi }}{\bigvee }\vdash _{\emptyset }\left(
C_{\xi }\vee D_{\xi ,j}\right) \right) $}

\item  \fbox{$f\!\left( \imath +1,X\right) \!=\!1\Leftrightarrow \underset{%
\xi \in \Xi }{\bigwedge }\left( \!\vdash _{\emptyset }\left( B_{\xi }\vee
X\right) \ or\ \!\underset{\text{ }j\in J_{\xi }}{\bigvee }f\!\left( \imath
,C_{\xi }\vee D_{\xi ,j}\right) \!={}\!1\right) $}
\end{enumerate}

Note that every ``$\,\vdash _{\emptyset }Y$ '' involved is expressible in
quantified boolean logic as $\forall x_{1}\cdots \forall x_{q}Y$, where $%
\left\{ x_{1},\cdots ,x_{q}\right\} $ is the set of propositional variables
occurring in $Y$. Having this, by recursion on $\imath $ with respect to
clauses 1--2 we obtain a desired ``transparent'' quantified boolean formula $%
\widehat{S}$ such that

\begin{equation*}
\fbox{$\mathbf{PDL}\vdash S$\ \ $\Leftrightarrow f\!\left(
s2^{s+t}+1,Z\right) =1\Leftrightarrow \mathbf{QBL}\vdash \widehat{S}$}
\end{equation*}
(\textbf{QBL} being\ the canonical proof system for quantified boolean
logic).

\begin{remark}
The size of $\widehat{S}$ is exponential in that of $S$, \footnote{%
{\footnotesize This is in contrast to analogous polynomial BCNE case, see
Corollary 20.}}\ whereas quantified boolean validity (and/or satisfiability)
is known to be PSPACE-complete (cf. e.g. \cite{Papa}). Hence $\mathbf{%
EXPTIME=PSPACE}$ holds if $\widehat{S}$ is equivalid with another quantified
boolean formula whose size is polynomial in the size of $S$. Moreover, this
holds true of $S:=\overline{\text{\textsc{Accepts}}}_{M,x}$ (see above and
Appendix B).
\end{remark}

\subsection{Conclusion}

Soundness and completeness together with full cut elimination [Theorems 4,
8] in semiformal (infinite) sequent calculus \textsc{Seq}$_{\omega }^{\text{%
\textsc{pdl}}}$ shows that Hilbert-Bernays-style proof system \textbf{PDL}
is a conservative extension of formal (finite) cutfree sequent calculi 
\textsc{Seq}$_{00}^{\text{\textsc{pdl}}}\varsubsetneq \,$\textsc{Seq}$_{0}^{%
\text{\textsc{pdl}}}$ and \textsc{Seq}$_{10}^{\text{\textsc{pdl}}%
}\varsubsetneq \,$\textsc{Seq}$_{1}^{\text{\textsc{pdl}}}$ with respect to
the corresponding classes of formulas \textrm{FOR}$_{00}\varsubsetneq \,$%
\textrm{FOR}$_{0}$ and \textrm{FOR}$_{10}\varsubsetneq \,$\textrm{FOR}$_{1}$%
, respectively. I.e., for any $A\in \,$\textrm{FOR}$_{\dagger }$, $\mathbf{%
PSP}\vdash A$ implies \textsc{Seq}$_{\dagger }^{\text{\textsc{pdl}}}\vdash A$
($\dagger \in \left\{ 0,1,00,10\right\} $). Here we let

\textrm{FOR}$_{1}$ := subset of \textrm{FOR} whose seq-formulas don't
include occurrences $\left[ P^{\ast }\right] $.

\textrm{FOR}$_{10}$ := subset of \textrm{FOR}$_{1}$ with atomic programs 
\textrm{PRO}$_{0}$.

\textrm{FOR}$_{0}$ := subset of \textrm{FOR}$_{1}$ whose seq-formulas don't
include occurrences $\left\langle P^{\ast }\right\rangle $, i.e. \textrm{FOR}%
$_{0}$ is just star-free fragment of \textrm{FOR}.

\textrm{FOR}$_{00}$ := subset of \textrm{FOR}$_{0}$ with atomic programs 
\textrm{PRO}$_{0}$.

Basic program connectives are interpretable in \textrm{FOR}$_{00}$ and 
\textrm{FOR}$_{10}$ by $\left( P;Q\right) \!A:=\left( P\right) \!\left(
Q\right) \!A$, $\left\langle P\cup Q\right\rangle \!A:=\left\langle
P\right\rangle \!A\vee \left\langle Q\right\rangle \!A$\ and $\left[ P\cup Q%
\right] \!A:=\left[ P\right] \!A\wedge \left[ Q\right] \!A$.

\textsc{Seq}$_{00}^{\text{\textsc{pdl}}}:=$%
\begin{equation*}
\begin{array}{c}
\quad 
\begin{array}{c}
\fbox{$\left( \text{\textsc{Ax}}\right) \quad x,\lnot x,\Gamma $}
\end{array}
\quad \quad \qquad \\ 
\begin{array}{c}
\fbox{$\left( \vee \right) \quad \dfrac{A,B,\Gamma }{A\vee B,\Gamma }\ $}%
\fbox{$\left( \wedge \right) \quad \dfrac{A,\Gamma \quad \quad B,\Gamma }{%
A\wedge B,\Gamma }\ $}
\end{array}
\\ 
\qquad \qquad \text{\ }\fbox{$
\begin{array}{c}
\left( \text{\textsc{Gen}}\right) \quad \dfrac{A_{1},\cdots ,A_{n}}{\left(
p\right) _{\chi _{1}}\!\!A_{1},\cdots ,\left( p\right) \!_{\chi
_{n}}\!A_{n},\Gamma }\ \left( n>0\right) \\ 
\text{if }\overset{n}{\underset{i=1}{\sum }}\chi _{i}=1.
\end{array}
$}\qquad \qquad \qquad
\end{array}
\end{equation*}

\textsc{Seq}$_{0}^{\text{\textsc{pdl}}}:=$%
\begin{equation*}
\begin{array}{c}
\begin{array}{c}
\fbox{$\left( \text{\textsc{Ax}}\right) \quad x,\lnot x,\Gamma $}
\end{array}
\quad \quad \qquad \\ 
\begin{array}{c}
\fbox{$\left( \vee \right) \quad \dfrac{A,B,\Gamma }{A\vee B,\Gamma }\ $}%
\fbox{$\left( \wedge \right) \quad \dfrac{A,\Gamma \quad \quad B,\Gamma }{%
A\wedge B,\Gamma }\ $}
\end{array}
\\ 
\ 
\begin{array}{c}
\fbox{$\left\langle \cup \right\rangle \quad \dfrac{\left\langle
P\right\rangle \!A,\left\langle R\right\rangle \!A,\Gamma }{\left\langle
P\cup R\right\rangle \!A,\Gamma }\ $}\fbox{$\left[ \cup \right] \quad \dfrac{%
\left[ P\right] \!A,\Gamma \text{\qquad }\left[ R\right] \!A,\Gamma }{\left[
P\cup R\right] \!A,\Gamma }\ $}
\end{array}
\, \\ 
\begin{array}{c}
\fbox{$\left\langle ;\right\rangle \quad \dfrac{\left\langle P\right\rangle
\!\!\left\langle R\right\rangle \!A,\Gamma }{\left\langle P;R\right\rangle
\!A,\Gamma }\ $}\fbox{$\left[ ;\right] \quad \dfrac{\left[ P\right] \!\left[
R\right] \!A,\Gamma }{\left[ P;R\right] \!A,\Gamma }\ $}
\end{array}
\\ 
\qquad \fbox{$
\begin{array}{c}
\left( \text{\textsc{Gen}}\right) \quad \dfrac{A_{1},\cdots ,A_{n}}{\left(
P\right) _{\chi _{1}}\!\!A_{1},\cdots ,\left( P\right) \!_{\chi
_{n}}\!A_{n},\Gamma }\ \left( n>0\right) \\ 
\text{if }\overset{n}{\underset{i=1}{\sum }}\chi _{i}=1.
\end{array}
$}\qquad \ 
\end{array}
\end{equation*}

\textsc{Seq}$_{10}^{\text{\textsc{pdl}}}:=$%
\begin{equation*}
\begin{array}{c}
\quad 
\begin{array}{c}
\fbox{$\left( \text{\textsc{Ax}}\right) \quad x,\lnot x,\Gamma $}
\end{array}
\quad \quad \qquad \\ 
\begin{array}{c}
\fbox{$\left( \vee \right) \quad \dfrac{A,B,\Gamma }{A\vee B,\Gamma }\ $}%
\fbox{$\left( \wedge \right) \quad \dfrac{A,\Gamma \quad \quad B,\Gamma }{%
A\wedge B,\Gamma }\ $}
\end{array}
\\ 
\fbox{$\left\langle \ast \right\rangle $\quad $\dfrac{\left\langle
p\right\rangle ^{m}\!\!A,\left\langle p^{\ast }\right\rangle \!A,\Gamma }{%
\!\!\left\langle p^{\ast }\right\rangle \!A,\Gamma }\left( m\geq 0\right) $}
\\ 
\qquad \qquad \text{\ }\fbox{$
\begin{array}{c}
\left( \text{\textsc{Gen}}\right) \quad \dfrac{A_{1},\cdots ,A_{n}}{\left(
p\right) _{\chi _{1}}\!\!A_{1},\cdots ,\left( p\right) \!_{\chi
_{n}}\!A_{n},\Gamma }\ \left( n>0\right) \\ 
\text{if }\overset{n}{\underset{i=1}{\sum }}\chi _{i}=1.
\end{array}
$}\qquad \qquad \qquad
\end{array}
\end{equation*}

\textsc{Seq}$_{1}^{\text{\textsc{pdl}}}:=$%
\begin{equation*}
\begin{array}{c}
\begin{array}{c}
\fbox{$\left( \text{\textsc{Ax}}\right) \quad x,\lnot x,\Gamma $}
\end{array}
\quad \quad \qquad \\ 
\begin{array}{c}
\fbox{$\left( \vee \right) \quad \dfrac{A,B,\Gamma }{A\vee B,\Gamma }\ $}%
\fbox{$\left( \wedge \right) \quad \dfrac{A,\Gamma \quad \quad B,\Gamma }{%
A\wedge B,\Gamma }\ $}
\end{array}
\\ 
\ 
\begin{array}{c}
\fbox{$\left\langle \cup \right\rangle \quad \dfrac{\left\langle
P\right\rangle \!A,\left\langle R\right\rangle \!A,\Gamma }{\left\langle
P\cup R\right\rangle \!A,\Gamma }\ $}\fbox{$\left[ \cup \right] \quad \dfrac{%
\left[ P\right] \!A,\Gamma \text{\qquad }\left[ R\right] \!A,\Gamma }{\left[
P\cup R\right] \!A,\Gamma }\ $}
\end{array}
\, \\ 
\begin{array}{c}
\fbox{$\left\langle ;\right\rangle \quad \dfrac{\left\langle P\right\rangle
\!\!\left\langle R\right\rangle \!A,\Gamma }{\left\langle P;R\right\rangle
\!A,\Gamma }\ $}\fbox{$\left[ ;\right] \quad \dfrac{\left[ P\right] \!\left[
R\right] \!A,\Gamma }{\left[ P;R\right] \!A,\Gamma }\ $}
\end{array}
\qquad \  \\ 
\fbox{$\left\langle \ast \right\rangle $\quad $\dfrac{\left\langle 
\overrightarrow{Q}\right\rangle \!\!\left\langle P\right\rangle
^{m}\!\!A,\left\langle \overrightarrow{Q}\right\rangle \!\!\left\langle
P^{\ast }\right\rangle \!A,\Gamma }{\left\langle \overrightarrow{Q}%
\right\rangle \!\!\left\langle P^{\ast }\right\rangle \!A,\Gamma }\left(
m\geq 0\right) $}\  \\ 
\fbox{$
\begin{array}{c}
\left( \text{\textsc{Gen}}\right) \quad \dfrac{A_{1},\cdots ,A_{n}}{\left(
P\right) _{\chi _{1}}\!\!A_{1},\cdots ,\left( P\right) \!_{\chi
_{n}}\!A_{n},\Gamma }\ \left( n>0\right) \\ 
\text{if }\overset{n}{\underset{i=1}{\sum }}\chi _{i}=1.
\end{array}
$}
\end{array}
\end{equation*}

It is readily seen that \textsc{Seq}$_{0}^{\text{\textsc{pdl}}}$ and \textsc{%
Seq}$_{00}^{\text{\textsc{pdl}}}$ (resp. \textsc{Seq}$_{1}^{\text{\textsc{pdl%
}}}$ and \textsc{Seq}$_{10}^{\text{\textsc{pdl}}}$) have the same provable
sequents modulo basic interpretation of \textrm{FOR}$_{0}$ (\textrm{FOR}$%
_{1} $) within \textrm{FOR}$_{00}$ (\textrm{FOR}$_{10}$).

As usual in proof theory, our cutfree sequent calculi provide useful help in
the verification of \textbf{PDL}-(un)provability of concrete formulas of
simple shapes. Concerning general computational complexities the following
holds.

1. The derivability (provability) in \textsc{Seq}$_{00}^{\text{\textsc{pdl}}%
} $ is a PSPACE problem. The case of \textsc{Seq}$_{0}^{\text{\textsc{pdl}}}$
is less clear as our interpretation of program connectives does not preserve
polynomial size.

2. The derivability (provability) in \textsc{Seq}$_{1}^{\text{\textsc{pdl}}}$
is EXPTIME-complete and in fact so is the derivability in \textsc{Seq}$%
_{10}^{\text{\textsc{pdl}}}$, too.

3. The latter is characteristic also for a subclass BDNE (: ``basic
disjunctive normal expressions'') of \textrm{FOR}$_{10}$, whereas \textsc{Seq%
}$_{10}^{\text{\textsc{pdl}}}$-derivability of dual BCNE expressions (:
``basic conjunctive normal expressions'') turns out to be decidable in
polynomial space. Moreover, \textsc{Seq}$_{10}^{\text{\textsc{pdl}}}$%
-derivability (and hence \textbf{PDL}-provability) of BDNE is equivalent to 
\textbf{QBL}-validity of corresponding ``transparent'' quantified boolean
formulas of exponential length.

Proofs of our claims use transfinite induction on predicative ordinal $%
\varphi _{\omega }\!\left( 0\right) $. It is not clear yet whether
conservative extension results (see above) are provable in Peano Arithmetic.

\subparagraph{Relevant papers}

\cite{HillPog} formalized \textbf{PDL} as proof system \textbf{CSPDL} with $%
\omega $-rule for $\left[ \ast \right] $ that is based on \emph{hypersequents%
} (more precisely: \emph{zoom tree hypersequents}), rather than sequents.
Sequent calculi proper are not exposed there. Consequently, conservative
extension corollaries and complexity connections are not mentioned, either.
It might appear, however, that the hypersequents were chosen to allow ``deep
inferences'', i.e. transformations applied to subformulas of those
explicitly shown (along the lines of Herbrand version \cite{HillPog} of
Gentzen's sequent calculus). Cut-elimination theorem is claimed for \textbf{%
CSPDL} but proof thereof is informal. Apparently it should follow by
Sch\"{u}tte-style predicative pattern adapted to the hypersequents, instead
of plain sequents. However no ordinal bounds on the height transformation is
given and hence proof theoretic strength of the required logic formalism
remains unclear. According to well-known predicative cut elimination in the
presence of $\omega $-rule(s), one would expect it to be the same as in the
present paper, provided that ordinal complexity of cut formulas with nested
occurrences of starred programs has the same natural upper bound $\omega
^{\omega }$ (which does not explicitly follow from \cite{HillPog}:
Definition 3.1).

\cite{HartLar} considered finite sequent calculus \textbf{GPDL} in an
extended language with mixed formulas (possibly including atomic programs)
and contextual sequents (whose antecedents and/or succedents might include
program terms). It is claimed that all but analytic cuts in special form can
be eliminated from \textbf{GPDL} derivations. There is no discussion of
possible conservative extensions and/or computational complexity connections.

\cite{GoreWid} presented a different approach in form of an optimal
tableau-based EXPTIME algorithm for deciding satisfiability for \textbf{PDL}
with converse (\textbf{CPDL}) without the use of analytic cuts. In order to
decide the satisfiability of a given input formula $\phi $ the algorithm
builds a suitable directed graph $G$ and checks the applicability of one of
the four attached rules $Rule\ 1$, ..., $Rule\ 4$. There is no obvious
translation into plain sequent calculus formalism.

\section{Appendix A: Ordinal assignments}

\subsection{Ordinal arithmetic}

We use basic properties 1--8 of Veblen's ordinals (abbr.:\negthinspace\ $%
\alpha $, $\!\beta $,\negthinspace\ $\gamma $,\negthinspace\ $\delta $) ( 
\cite{Veblen},\negthinspace\ \cite{Fef2},\negthinspace\ \cite{Poh}).

\begin{enumerate}
\item  Basic relation $<$ is linearly ordered.

\item  Symmetric sum is associative and commutative.

\item  $0<1=\omega ^{0}$, $\omega =\omega ^{1}$, $\omega ^{\beta }=\varphi
\left( 0,\beta \right) $.

\item  $\alpha +\!\!\!\!\!+\,0=\alpha $, $\alpha <\beta \rightarrow \alpha
+\!\!\!\!\!+\gamma <\beta +\!\!\!\!\!+\gamma +\!\!\!\!\!+\delta $.

\item  $\alpha <\beta \rightarrow \varphi \left( \alpha ,\gamma \right)
<\varphi \left( \beta ,\gamma \right) \wedge \varphi \left( \gamma ,\alpha
\right) <\varphi \left( \gamma ,\beta \right) $.

\item  $\alpha \leq \beta <\varphi \left( \gamma ,\delta \right) \rightarrow
\alpha +\!\!\!\!\!+\beta <\varphi \left( \gamma ,\delta \right) $.

\item  $\alpha <\beta \wedge \gamma <\varphi \left( \beta ,\delta \right)
\rightarrow \varphi \left( \alpha ,\gamma \right) <\varphi \left( \beta
,\delta \right) \wedge \varphi \left( \alpha ,\varphi \left( \beta ,\delta
\right) \right) =\varphi \left( \beta ,\delta \right) $.

\item  $\alpha \leq \omega ^{\alpha }$, $0<\alpha \rightarrow \omega
^{\varphi \left( \alpha ,\beta \right) }=\varphi \left( \alpha ,\beta
\right) $.
\end{enumerate}

$\varphi \left( \alpha ,\beta \right) $ is also denoted by $\varphi _{\alpha
}\!\left( \beta \right) $. Note that $\varepsilon _{0}=\varphi _{1}\left(
0\right) <\varphi _{\omega }\!\left( 0\right) <\Gamma _{0}$.

In the rest of this chapter we freely use these properties without explicit
references.

\subsection{Cut elimination $\partial \hookrightarrow \mathcal{E}\left(
\partial \right) $}

For the sake of brevity we'll slightly refine our inductive definition of $%
\mathcal{E\!}\left( \partial \right) $. To this end we upgrade $\mathcal{R}$
to $\mathcal{R}^{+}:\left( \partial \mid \!\!\!\frac{\beta }{\,\rho +\omega
^{a}}\,\Delta \right) \hookrightarrow \left( \mathcal{R}^{+}\mathcal{\!}%
\left( \rho ,\alpha ,\partial \right) \mid \!\!\!\frac{\,\varphi \left(
\alpha ,\beta \right) }{\,\rho }\,\Delta \right) $. That is, for any $\rho
>0 $, $\alpha $ and $\left( \partial :\Delta \right) $\ with $\deg \left(
\partial \right) <\rho +\omega ^{a}$ we define $\left( \mathcal{R}^{+}%
\mathcal{\!}\left( \rho ,\alpha ,\partial \right) :\Delta \right) $ such
that $\deg \left( \mathcal{R}^{+}\mathcal{\!}\left( \rho ,\alpha ,\partial
\right) \right) <\rho $ and $h\left( \mathcal{R}^{+}\mathcal{\!}\left( \rho
,\alpha ,\partial \right) \right) <$ $\varphi \left( \alpha ,h\left(
\partial \right) \right) $. Then for any $\partial $ with cuts\ we let 
\begin{equation*}
\fbox{$\mathcal{E}\left( \partial \right) :=\mathcal{R}^{+}\mathcal{\!}%
\left( 1,\alpha ,\partial \right) $, where $\alpha :=\min \left\{ \beta
:\deg \left( \partial \right) <\omega ^{\beta }\right\} $}
\end{equation*}
and conclude that $\deg \left( \mathcal{E}\left( \partial \right) \right) =0$
and $h\left( \mathcal{E}\left( \partial \right) \right) <\varphi \left(
\alpha ,h\left( \partial \right) \right) $.

Now $\mathcal{R}^{+}\mathcal{\!}\left( \rho ,\alpha ,\partial \right) $ is
defined for any $\partial $ with $\deg \left( \partial \right) <\rho +\omega
^{a}$ as follows by double induction on $\alpha $ and $h\left( \partial
\right) $. Let $\left( \text{\textsc{R}}\right) $ be the lowermost inference
in $\partial $. If $\left( \text{\textsc{R}}\right) $ is not a $\left( \text{%
\textsc{Cut}}\right) $ on $C$ with $\frak{o\!}\left( C\right) +1\geq \rho $
then $\mathcal{R}^{+}\mathcal{\!}\left( \rho ,\alpha ,\partial \right) $
arises from $\partial $ by substituting $\mathcal{R}^{+}\mathcal{\!}\left(
\rho ,\alpha ,\partial _{i}\right) $ for the lowermost subdeductions $%
\partial _{i}$ (recall that $h\left( \partial _{i}\right) <h\left( \partial
\right) $ ). Otherwise, we have 
\begin{equation*}
\left( \partial :\Gamma \cup \Pi \right) =\dfrac{\left( \partial
_{1}:C,\Gamma \right) \quad \quad \left( \partial _{2}:\overline{C},\Pi
\right) }{\Gamma \cup \Pi }\ \left( \text{\textsc{Cut}}\right)
\end{equation*}
where $\rho \leq \frak{o\!}\left( C\right) +1\leq \deg \left( \partial
\right) <\rho +\omega ^{a}$. Let 
\begin{equation*}
\left( \widehat{\partial }:\Gamma \cup \Pi \right) :=\dfrac{\left( \mathcal{R%
}^{+}\left( \rho ,\alpha ,\partial _{1}\right) :C,\Gamma \right) \quad \quad
\left( \mathcal{R}^{+}\mathcal{\!}\left( \rho ,\alpha ,\partial _{2}\right) :%
\overline{C},\Pi \right) }{\Gamma \cup \Pi }\ \left( \text{\textsc{Cut}}%
\right)
\end{equation*}
and consider two cases.

\subparagraph{Case $\protect\alpha =0$.}

Let $\mathcal{R}^{+}\mathcal{\!}\left( \rho ,\alpha ,\partial \right) =%
\mathcal{R}^{+}\mathcal{\!}\left( \rho ,0,\partial \right) :=\mathcal{R}%
\left( \widehat{\partial }\right) $. Recall that 
\begin{equation*}
\deg \left( \mathcal{R}\left( \widehat{\partial }\right) \right) <\deg
\left( \widehat{\partial }\right) =\frak{o\!}\left( C\right) +1\leq \deg
\left( \partial \right) <\rho +1
\end{equation*}
and hence $\deg \left( \mathcal{R}^{+}\mathcal{\!}\left( \rho ,\alpha
,\partial \right) \right) =\deg \left( \mathcal{R}\left( \widehat{\partial }%
\right) \right) <\rho $. On the other hand 
\begin{eqnarray*}
h\left( \mathcal{R}\left( \widehat{\partial }\right) \right) \! &<&\!h\left( 
\mathcal{R}^{+}\mathcal{\!}\left( \rho ,\alpha ,\partial _{1}\right) \right)
+\!\!\!\!\!+\,\,h\left( \mathcal{R}^{+}\mathcal{\!}\left( \rho ,\alpha
,\partial _{2}\right) \right) +\omega \\
&\leq &\omega ^{\!h\left( \mathcal{R}^{+}\mathcal{\!}\left( \rho ,\alpha
,\partial _{1}\right) \right) }+\!\!\!\!\!+\,\,\omega ^{h\left( \mathcal{R}%
^{+}\mathcal{\!}\left( \rho ,\alpha ,\partial _{2}\right) \right) }+\omega \\
&<&\omega ^{\!h\left( \widehat{\partial }\right) }=\varphi \left( 0,h\left( 
\widehat{\partial }\right) \right)
\end{eqnarray*}
which yields $h\left( \mathcal{R}^{+}\mathcal{\!}\left( \rho ,\alpha
,\partial \right) \right) \!=\!h\left( \mathcal{R}\left( \widehat{\partial }%
\right) \right) \!<\!\varphi \left( 0,h\left( \widehat{\partial }\right)
\right) \!$, as desired.

\subparagraph{Case $\protect\alpha >0$.}

Thus $\omega ^{a}=\omega ^{\alpha _{1}}+\cdots +\omega ^{\alpha _{n}}$ for $%
\alpha >\alpha _{1}\geq \cdots \geq \alpha _{n}$ (by Cantor's normal form).
In this case we apply inductive hypotheses successively for $\alpha
_{1},\cdots ,\alpha _{n}$ and let 
\begin{equation*}
\mathcal{R}^{+}\mathcal{\!}\left( \rho ,\alpha ,\partial \right) :=\mathcal{R%
}^{+}\mathcal{\!}\left( \rho ,\alpha _{1},\mathcal{R}^{+}\mathcal{\!}\left(
\rho _{1},\alpha _{2},\mathcal{R}^{+}\left( \cdots ,\mathcal{R}^{+}\mathcal{%
\!}\left( \rho _{n-1},\alpha _{n},\partial \right) \right) \right) \right) 
\end{equation*}
where $\rho _{0}:=\rho $ and $\left( \forall i>0\right) \rho _{i+1}:=\rho
_{i}+\omega ^{\alpha _{i+1}}$. Then $\deg \left( \mathcal{R}^{+}\mathcal{\!}%
\left( \rho ,\alpha ,\partial \right) \right) <\rho $ and $h\left( \mathcal{R%
}^{+}\mathcal{\!}\left( \rho ,\alpha ,\partial \right) \right) <\varphi
\left( \alpha _{1},\varphi \left( \alpha _{2},\varphi \left( \cdots ,\varphi
\left( \alpha _{n},h\left( \partial \right) \right) \right) \right) \right)
<\varphi \left( \alpha ,h\left( \partial \right) \right) $, as desired.

\subsection{Formalization}

We fix a chosen ``canonical'' primitive recursive ordinal representation 
\begin{equation*}
\fbox{$\mathcal{O}=\left\langle 0,1,\omega ,<,+,+\!\!\!\!+,\omega ^{\left(
-\right) },\varphi \left( -,-\right) \right\rangle $}
\end{equation*}
(also known as \emph{system of ordinal notations}) in the language of $%
\mathbf{PA}$ that is supposed to be well-ordered by $<$ up to $\varphi
_{\omega }\left( 0\right) $ (at least). To formalize the latter assumption
we extend\ standard formalism of $\mathbf{PA}$ by the transfinite induction
axiom (schema) for arbitrary arithmetical formulas, $\mathrm{TI}_{\mathcal{O}%
}\left( \varphi _{\omega }\!\left( 0\right) \right) $. The extended proof
system is abbreviated by $\mathbf{PA}_{\varphi _{\omega }\left( 0\right) }$.
Derivations $\partial $ used in the proofs are interpreted as primitive
recursive trees whose nodes $x$ are labeled with sequents and ordinals $%
ord\left( x\right) <\varphi _{\omega }\!\left( 0\right) $. Having this it is
easy to formalize in $\mathbf{PA}_{\varphi _{\omega }\left( 0\right) }$ the
whole cut elimination proof; note that the operators $\mathcal{R}$, $%
\mathcal{R}^{+}$\ and $\mathcal{E}$\ involved are constructively defined and 
$\mathrm{TI}_{\mathcal{O}}\left( \varphi _{\omega }\!\left( 0\right) \right) 
$ is used in the corresponding termination-and-correctness proofs only.
Actually we can restrict $\mathrm{TI}_{\mathcal{O}}\left( \varphi _{\omega
}\!\left( 0\right) \right) $ to primitive recursive induction formulas thus
reducing $\mathbf{PA}_{\varphi _{\omega }\left( 0\right) }$ to $\mathbf{PRA}%
_{\varphi _{\omega }\left( 0\right) }$.

\section{Appendix B: Formula \textsc{Accepts}$_{M,x}$ \protect\footnote{%
{\footnotesize This is a recollection of \cite{DynamicLogic}: 8.2.}}}

\subsection{Semantics}

Consider a given polynomial-space-bounded $k$-tape alternating Turing
machine $M$ on a given input $x$ of length $n$ with blanks over $M$'s input
alphabet; $\vdash $ and $\dashv $ are the left and right endmarkers,
respectively. Formula \textsc{Accepts}$_{M,x}$ involves the single atomic
program \textsc{Next}, atomic propositions \textsc{Symbol}$_{i}^{a}$ and 
\textsc{State}$_{i}^{q}$ for each symbol $a$ in $M$'s tape alphabet, $q$ a
state of $M$'s finite control, and $0\leq i\leq n$, and an atomic
proposition \textsc{Accept}. Then \textsc{Accepts}$_{M,x}$ has the property
that any satisfying Kripke frame encodes an accepting computation of $M$ on $%
x$. In any such Kripke frame, states $u$ represent configurations of $M$
occurring in the computation tree of $M$ on input $x=x_{1},\cdots ,x_{n}$;
the truth values of \textsc{Symbol}$_{i}^{a}$ and \textsc{State}$_{i}^{q}$
at state $u$ give the tape contents, current state, and tape head position
in the configuration corresponding to $u$. The truth value of \textsc{Accept}
will be $\mathbf{1}$ iff the computation beginning in state $u$ is an
accepting computation according to the rules of alternating Turing machine
acceptance. Then $M$ accepts $x$ iff \textsc{Accepts}$_{M,x}$ is
satisfiable. \textsc{Accepts}$_{M,x}$ is EXPTIME-complete (cf. \cite
{DynamicLogic}: Theorem 8.5) and hence so is the negation $\overline{\text{%
\textsc{Accepts}}}_{M,x}$.

\subsection{Formal definition}

Let $\Gamma $ be $M$'s tape alphabet and $Q$ the set of states; there is a
distinguished start-state $s\in Q$ and left/right annotations $\ell ,r\notin
Q$. Let $U\subseteq Q$ \ and $E\subseteq Q$ be the sets of universal and
existential states, respectively. Thus $U\cup E=Q$ and $U\cap E=\emptyset $.
For each pair $\left( q,a\right) \in Q\times \Gamma $ let $\Delta \left(
q,a\right) $ be the set of all triples describing a possible action when
scanning $a$ in state $q$. Working in $\mathcal{L}$ we let 
\begin{equation*}
\fbox{\textsc{Accepts}$_{M,x}:=\text{\textsc{Acc}}\wedge \!\text{\textsc{%
Start\negthinspace }}\wedge \!\left[ \text{\textsc{Next*}}\right] \!\left( 
\text{\textsc{Config\negthinspace }}\wedge \!\text{\textsc{Move\negthinspace 
}}\wedge \!\text{\textsc{Acceptance}}\right) $}
\end{equation*}
where \textsc{Acc(ept), State}$_{\left( -\right) }^{\left( -\right) }$%
\textsc{, Symbol}$_{\left( -\right) }^{\left( -\right) }\in V\!AR$, \textsc{%
Next\thinspace }$\in PRO$ \smallskip while \textsc{Start}, \textsc{Config}, 
\textsc{Move} and \textsc{Acceptance} are defined as follows.

1. \textsc{Start }$:=$ \textsc{State}$_{0}^{s}\wedge \underset{1\leq i\leq n%
}{\bigwedge }$\textsc{Symbol}$_{i}^{x_{i}}\wedge \underset{n+1\leq i\leq
n^{k}}{\bigwedge }$\textsc{Symbol}$_{i}^{\Box }$.

2. \textsc{Config }$:=$

$\underset{0\leq i\leq n+1}{\bigwedge }\underset{a\in \Gamma }{\bigvee }%
\!\left( \!\text{\textsc{Symbol}}_{i}^{a}\!\wedge \underset{a\neq b\in
\Gamma }{\bigwedge }\overline{\text{\textsc{Symbol}}}_{i}^{b}\!\right)
\!\wedge $ \textsc{Symbol}$_{0}^{\vdash }\wedge \overline{\text{\textsc{%
Symbol}}}_{n+1}^{\dashv }\wedge $

$\underset{0\leq i\leq n+1}{\bigvee }\underset{q\in Q}{\bigvee }$\textsc{%
State}$_{i}^{q}\wedge \underset{0\leq i\leq n+1}{\bigwedge }\underset{q\in
Q\cup \left\{ \ell ,r\right\} }{\bigvee }\left( \text{\textsc{State}}%
_{i}^{q}\wedge \underset{q\neq p\in Q\cup \left\{ l,r\right\} }{\bigwedge }%
\overline{\text{\textsc{State}}}_{i}^{p}\right) \wedge $

$\underset{0\leq i\leq n}{\bigwedge }\underset{q\in Q\cup \left\{ \ell
\right\} }{\bigwedge }\!\left( \overline{\text{\textsc{State}}}_{i}^{q}\vee 
\text{\textsc{State}}_{i+1}^{\ell }\right) \wedge \underset{0\leq i\leq n+1}{%
\bigwedge }\underset{q\in Q\cup \left\{ r\right\} }{\bigwedge }\!\left( 
\overline{\text{\textsc{State}}}_{i}^{q}\vee \text{\textsc{State}}%
_{i-1}^{r}\right) $.\smallskip \smallskip

3. \textsc{Move }$:=$

$\underset{0\leq i\leq n+1}{\bigwedge }\left( \overline{\text{\textsc{State}}%
}_{i}^{\ell }\vee \overline{\text{\textsc{State}}}_{i}^{r}\vee \underset{%
a\in \Gamma }{\bigwedge }\left( \overline{\text{\textsc{Symbol}}}%
_{i}^{a}\vee \left[ \text{\textsc{Next}}\right] \text{\textsc{Symbol}}%
_{i}^{a}\right) \right) \wedge $

$\underset{0\leq i\leq n+1}{\bigwedge }\underset{\QATOP{a\in \Gamma }{q\in Q}%
}{\bigwedge }\left( 
\begin{array}{c}
\overline{\text{\textsc{Symbol}}}_{i}^{a}\vee \overline{\text{\textsc{State}}%
}_{i}^{q}\vee \\ 
\left( 
\begin{array}{c}
\underset{\left( p,b,d\right) \in \Delta \left( q,a\right) }{\bigwedge }%
\left\langle \text{\textsc{Next}}\right\rangle \left( \text{\textsc{Symbol}}%
_{i}^{b}\wedge \text{\textsc{State}}_{i+d}^{p}\right) \wedge \\ 
\left[ \text{\textsc{Next}}\right] \left( \underset{\left( p,b,d\right) \in
\Delta \left( q,a\right) }{\bigvee }\left( \text{\textsc{Symbol}}%
_{i}^{b}\wedge \text{\textsc{State}}_{i+d}^{p}\right) \right)
\end{array}
\right)
\end{array}
\right) $.\smallskip

4. \textsc{Acceptance }$:=$

$\left( \underset{0\leq i\leq n+1}{\bigwedge }\underset{q\in E}{\bigwedge }%
\overline{\text{\textsc{State}}}_{i}^{q}\vee \left( \left( \text{\textsc{Acc}%
}\vee \left[ \text{\textsc{Next}}\right] \overline{\text{\textsc{Acc}}}%
\right) \wedge \left( \overline{\text{\textsc{Acc}}}\vee \left\langle \text{%
\textsc{Next}}\right\rangle \text{\textsc{Acc}}\right) \right) \right)
\wedge $

$\left( \underset{0\leq i\leq n+1}{\bigwedge }\underset{q\in U}{\bigwedge }%
\overline{\text{\textsc{State}}}_{i}^{q}\vee \left( \left( \text{\textsc{Acc}%
}\vee \left\langle \text{\textsc{Next}}\right\rangle \overline{\text{\textsc{%
Acc}}}\right) \wedge \left( \overline{\text{\textsc{Acc}}}\vee \left[ \text{%
\textsc{Next}}\right] \text{\textsc{Acc}}\right) \right) \right) $.\smallskip

Hence 
\begin{equation*}
\fbox{$\left. 
\begin{array}{c}
\overline{\text{\textsc{Accepts}}}_{M,x}=\smallskip \\ 
\overline{\text{\textsc{Acc}}}\vee \overline{\text{\textsc{Start}}}\vee
\!\left\langle \text{\textsc{Next*}}\right\rangle \!\left( \overline{\text{%
\textsc{Config}}}\vee \!\overline{\text{\textsc{Move}}}\vee \!\overline{%
\text{\textsc{Acceptance}}}\right)
\end{array}
\right. $}
\end{equation*}
is equivalent to \fbox{$\left\langle p^{\ast }\right\rangle \!A\vee Z$} for $%
\fbox{$p=\ $\textsc{Next}, $Z=\overline{\text{\textsc{Acc}}}\vee \overline{%
\text{\textsc{Start}}}\in \mathcal{L}_{00}$}$ and

\begin{equation*}
\fbox{$\left. 
\begin{array}{c}
A=F_{0}\vee \left( F_{1}\wedge \left[ p\right] G_{1}\right) \vee \left(
F_{2}\wedge \left[ p\right] G_{2}\right) \vee \underset{\alpha \in R}{%
\bigvee }\left( F_{\alpha }\wedge \left[ p\right] G_{\alpha }\right) \vee
\left( F_{3}\wedge \left\langle p\right\rangle G_{3}\right) \\ 
\vee \left( F_{4}\wedge \left\langle p\right\rangle G_{4}\right) \vee 
\underset{\beta \in T}{\bigvee }\left( F_{\beta }\wedge \left\langle
p\right\rangle G_{\beta }\right) \vee \underset{\gamma \in S}{\bigvee }%
\left( F_{\gamma }\wedge \left\langle p\right\rangle G_{\gamma }\right) \in 
\mathrm{BDNF}
\end{array}
\right. $}
\end{equation*}

where:

$R=\left\{ \alpha =\left( i,a,q,\left( p,b,d\right) \right) \in \left[ n+1%
\right] \times \Gamma \times Q\times \Delta \left( q,a\right) \right\} $,

$T=\left\{ \beta =\left( i,a\right) \in \left[ n+1\right] \times \Gamma
\right\} $,

$S=\left\{ \gamma =\left( i,a,q\right) \in \left[ n+1\right] \times \Gamma
\times Q\right\} $,

$F_{0}=\underset{0\leq i\leq n+1}{\bigvee }\underset{a\in \Gamma }{\bigwedge 
}\!\left( \!\overline{\text{\textsc{Symbol}}}_{i}^{a}\!\vee \underset{a\neq
b\in \Gamma }{\bigvee }\text{\textsc{Symbol}}_{i}^{b}\!\right) \!\vee 
\overline{\text{\textsc{Symbol}}}_{0}^{\vdash }\vee \,$\textsc{Symbol}$%
_{n+1}^{\dashv }$

$\vee \underset{0\leq i\leq n+1}{\bigwedge }\underset{q\in Q}{\bigwedge }%
\overline{\text{\textsc{State}}}_{i}^{q}\vee \underset{0\leq i\leq n+1}{%
\bigvee }\underset{q\in Q\cup \left\{ \ell ,r\right\} }{\bigwedge }\left( 
\overline{\text{\textsc{State}}}_{i}^{q}\vee \underset{q\neq p\in Q\cup
\left\{ l,r\right\} }{\bigvee }\text{\textsc{State}}_{i}^{p}\right) \vee $

$\underset{0\leq i\leq n}{\bigvee }\underset{q\in Q\cup \left\{ \ell
\right\} }{\bigvee }\!\left( \overline{\text{\textsc{State}}}_{i}^{q}\wedge 
\text{\textsc{State}}_{i+1}^{\ell }\right) \vee \underset{0\leq i\leq n+1}{%
\bigvee }\underset{q\in Q\cup \left\{ r\right\} }{\bigvee }\!\left( 
\overline{\text{\textsc{State}}}_{i}^{q}\wedge \text{\textsc{State}}%
_{i-1}^{r}\right) $,

\smallskip $F_{1}=\underset{0\leq i\leq n+1}{\bigwedge }\underset{q\in E}{%
\bigwedge }\overline{\text{\textsc{State}}}_{i}^{q}\wedge \ $\textsc{Acc}, $%
G_{1}=\overline{\text{\textsc{Acc}}}$,

$F_{2}=\underset{0\leq i\leq n+1}{\bigwedge }\underset{q\in U}{\bigwedge }%
\overline{\text{\textsc{State}}}_{i}^{q}\wedge \overline{\text{\textsc{Acc}}}
$, $G_{2}=\ $\textsc{Acc},

$F_{3}=$\textsc{\thinspace }$\underset{0\leq i\leq n+1}{\bigwedge }\underset{%
q\in E}{\bigwedge }\overline{\text{\textsc{State}}}_{i}^{q}\wedge \ 
\overline{\text{\textsc{Acc}}}$, $G_{3}=\ $\textsc{Acc},

$F_{4}=\underset{0\leq i\leq n+1}{\bigwedge }\underset{q\in U}{\bigwedge }%
\overline{\text{\textsc{State}}}_{i}^{q}\wedge $\textsc{Acc}, $G_{4}=%
\overline{\text{\textsc{Acc}}}$,

$F_{\alpha }=\ $\textsc{Symbol}$_{i}^{a}\wedge \overline{\text{\textsc{State}%
}}_{i}^{q}$, $G_{\alpha }=\overline{\text{\textsc{Symbol}}}_{i}^{b}\vee 
\overline{\text{\textsc{State}}}_{i+d}^{p}$,

$F_{\beta }=\ $\textsc{State}$_{i}^{\ell }\wedge \ $\textsc{State}$%
_{i}^{r}\wedge \ $\textsc{Symbol}$_{i}^{a}$, $G_{\beta }=\overline{\text{%
\textsc{Symbol}}}_{i}^{a}$,

$F_{\gamma }=\ $\textsc{Symbol}$_{i}^{a}\wedge \overline{\text{\textsc{State}%
}}_{i}^{q}$, $G_{\gamma }=\underset{\left( p,b,d\right) \in \Delta \left(
q,a\right) }{\bigvee }\left( \overline{\text{\textsc{Symbol}}}_{i}^{b}\vee 
\overline{\text{\textsc{State}}}_{i+d}^{p}\right) $.$\ $

Note that $\left| \left\langle p^{\ast }\right\rangle \!A\vee Z\right| $ is
at most quadratic in $\left| \text{\textsc{Accepts}}_{M,x}\right| $.

------------------------------------------------------------------------------------------

\end{document}